\documentclass{vldb}
\usepackage{multirow}

\usepackage{booktabs} 
\usepackage{times} 
\usepackage{color} 
\usepackage{balance} 
\usepackage{url} 
\usepackage{tabularx}
\usepackage{siunitx} 
\usepackage{epsfig,tabularx,subfigure,multirow, graphicx}
\usepackage{enumitem}

\usepackage{hyperref} 

\usepackage{amsthm,amssymb,amsmath}  

\newcolumntype{L}[1]{>{\raggedright\arraybackslash}p{#1}}
\newcolumntype{C}[1]{>{\centering\arraybackslash}p{#1}}
\newcolumntype{R}[1]{>{\raggedleft\arraybackslash}p{#1}}

\usepackage[linesnumbered,ruled,vlined]{algorithm2e}
\SetKwRepeat{Do}{do}{while}
\SetCommentSty{mycommfont}

\long\def\comment#1{}

\newcommand{\nop}[1]{}

\newcommand{\figureCaptionMargin}{\vspace{-2ex}}
\newcommand{\figureBelowMargin}{\vspace{-2ex}}

\newtheorem{theorem}{\bf Theorem}[section]
\newtheorem{lemma}{\bf Lemma}[section]

\newtheorem{example}{\bf Example}

\theoremstyle{remark}

\theoremstyle{definition}
\newtheorem{definition}{\bf Definition}

\vldbTitle{CoinMagic: A Differential Privacy Framework for Ring Signature Schemes}
\vldbAuthors{Wangze Ni, Han Wu, Peng Cheng et al}
\vldbDOI{https://doi.org/xxx.xxxx/xxx.xxxx}

\title{CoinMagic: A Differential Privacy Framework for Ring Signature Schemes}

\author{
	{Wangze Ni{\small $~^{\dagger}$}, Han Wu{\small $~^{*}$}, Peng Cheng{\small $~^{*}$}, Lei Chen{\small $~^{\dagger}$}, Xuemin Lin{\small $~^{\#}$}, 
	Lei Chen{\small $~^{\ddagger}$}, Xin Lai{\small $~^{\ddagger}$}, Xiao Zhang{\small $~^{\ddagger}$}
	} \\
	\fontsize{10}{10}\selectfont\itshape
	$~^{\dagger}$The Hong Kong University of Science and Technology, Hong Kong, China\\
	\fontsize{9}{9}\selectfont\ttfamily\upshape
	\{wniab, leichen\}@cse.ust.hk \\
	\fontsize{10}{10}\selectfont\itshape
	$~^{*}$East China Normal University, Shanghai, China\\
	\fontsize{9}{9}\selectfont\ttfamily\upshape
	han.wu@stu.ecnu.edu.cn, pcheng@sei.ecnu.edu.cn \\
	\fontsize{10}{10}\selectfont\itshape
	$~^{\#}$The University of New South Wales, Australia\\
	\fontsize{9}{9}\selectfont\ttfamily\upshape
	lxue@cse.unsw.edu.au\\
	\fontsize{10}{10}\selectfont\itshape
	$~^{\ddagger}$ Shenzhen Onething Technologies Co., Ltd., Shenzhen, China\\
	\fontsize{9}{9}\selectfont\ttfamily\upshape
	\{leichen, laixin, zhangxiao\}@onething.net \\
}

\begin{document}

	\maketitle

	\begin{abstract}

By allowing users to obscure their transactions via including ``mixins"  (chaff coins), ring signature schemes have been widely used to protect a sender's identity of a transaction in privacy-preserving blockchain systems, like Monero and Bytecoin. 
 However, recent works point out that the existing ring signature scheme is vulnerable to the ``chain-reaction" analysis (i.e., the spent coin in a given ring signature can be deduced through elimination). Especially, when the diversity of mixins is low, the spent coin will have a high risk to be detected. To overcome the weakness, the ring signature should be consisted of a set of mixins with high diversity and produce observations having ``similar'' distributions for any two coins. In this paper, we propose a notion, namely $\epsilon$-coin-indistinguishability ($\epsilon$-CI), to formally define the ``similar'' distribution guaranteed through a differential privacy scheme. Then, we formally define the CI-aware mixins selection problem with disjoint-superset constraint (CIA-MS-DS), which aims to find a mixin set that has maximal diversity and satisfies the constraints of $\epsilon$-CI and the budget. In CIA-MS-DS, each ring signature is either disjoint with or the superset of its preceding ring signatures. We prove that CIA-MS-DS is NP-hard and thus intractable. To solve the CIA-MS-DS problem, we propose two approximation algorithms, namely the Progressive Algorithm and the Game Theoretic Algorithm, with theoretic guarantees. Through extensive experiments on both real data sets and synthetic data sets, we demonstrate the efficiency and the effectiveness of our approaches.

\end{abstract}

	\section{Introduction}

Recently, with the success of cryptocurrencies (e.g., Bitcoin~\cite{nakamoto2008bitcoin}, Ethereum~\cite{wood2014ethereum}) and blockchain products (e.g., Thunderchain~\cite{linkchain}), blockchain technologies have attracted much attention from both academia (e.g., database community~\cite{xu2019vchain}~\cite{zhang2019gem}) and industry (e.g., supply chain management~\cite{korpela2017digital}, healthcare~\cite{azaria2016medrec} and bank~\cite{guo2016blockchain}). In general, blockchain is a secure data structure maintained by untrusted peers in a decentralized P2P network. It has many valuable features such as transparency, provenance, fault tolerance, and authenticity.

\begin{figure}[t!]\vspace{-1ex}
	\centering
	\includegraphics[scale=0.35]{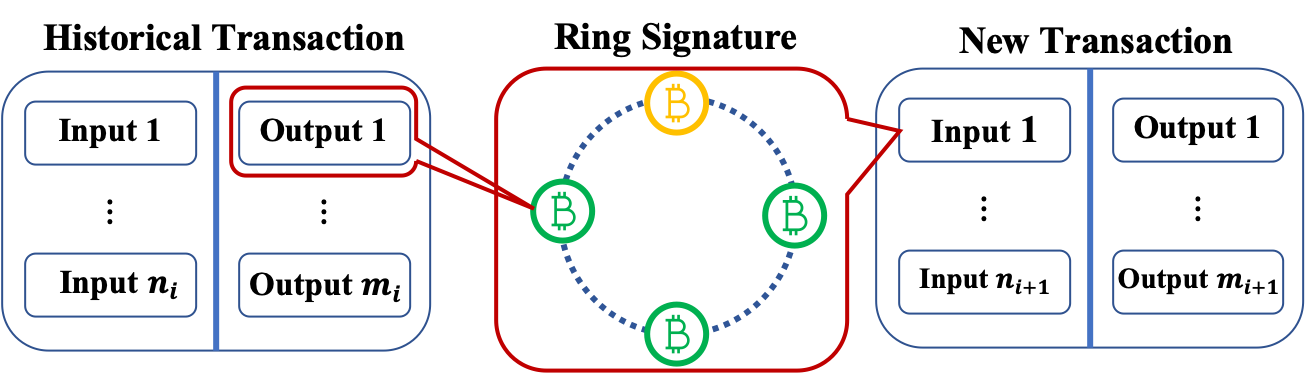}\vspace{-3ex}
	\caption{\small The Unspent Transaction Output (UTXO) model}
	\label{fig:coins}\vspace{-5ex}
\end{figure}

\begin{figure*}[t]\vspace{-4ex}
	\subfigure[][{The permutation with $\tau_1$}]{
		\scalebox{0.45}[0.4]{\includegraphics{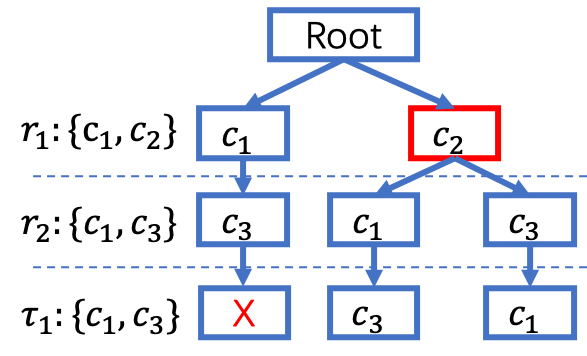}}
		\label{fig:solution1}}
	\subfigure[][{The permutation with $\tau_2$}]{
		\scalebox{0.4}[0.4]{\includegraphics{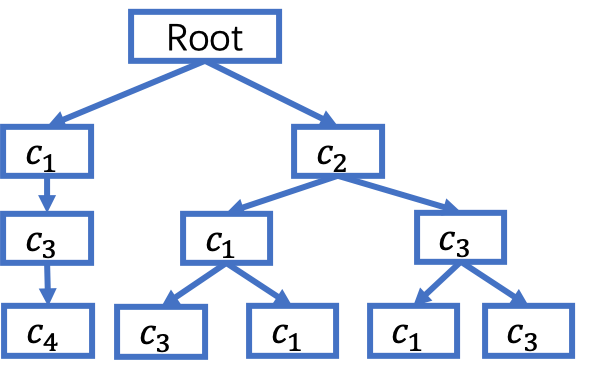}}
		\label{fig:solution2}}
	\subfigure[][{The permutation with $\tau_3$}]{
		\scalebox{0.35}[0.4]{\includegraphics{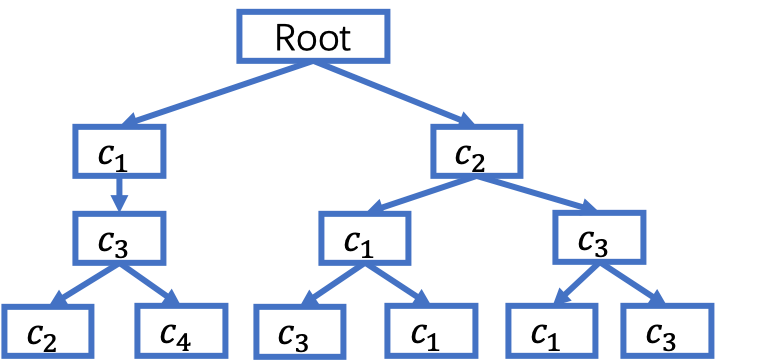}}
		\label{fig:solution3}}
	\subfigure[][{The permutation with $\tau_4$}]{
		\scalebox{0.45}[0.4]{\includegraphics{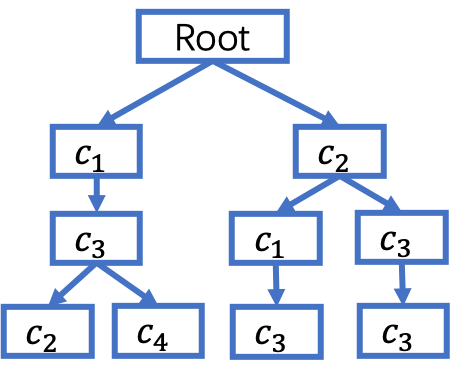}}
		\label{fig:solution4}}
	\vspace{-2.5ex}
	\caption{Motivation Example.}\vspace{-4.5ex}
	\label{fig:motivation}
\end{figure*}

However, transparency property of blockchain can initiate privacy problem, that is, in many real-world applications where users want to keep the transaction information by themselves.
For instance, in a trading system, a trader may hope to conceal the information of her/his trade partner who received (sent) money from (to) her/him. Furthermore, the transaction information can be easily linked to a variety of other information that an individual usually wishes to protect. For example, in a blockchain-based ride-hailing system, like MVL~\cite{MVL}, there are massive users' location data. By collecting and processing transaction data, it is possible to infer the user's personal information such as addresses of home~\cite{andres2012geo}.

To protect privacy, researchers have proposed some privacy definitions~\cite{okamoto1991universal}~\cite{van2013cryptonote} and privacy-preserving methods~\cite{van2013cryptonote}~\cite{sun2017ringct}~\cite{yuen2019ringct}. In~\cite{okamoto1991universal}, T. Okamoto and K. Ohta introduced that ``privacy'' must be one criterion of ideal electronic cash, which means  ``the relationship between the user and his purchases must be untraceable by anyone". Van Saberhagen proposed that ``Untraceability'' must be satisfied for a fully anonymous electronic cash model~\cite{van2013cryptonote}, which refers to ``for each incoming transaction all possible senders are equiprobable". To satisfy the untraceability, \textit{ring signature} (RS) schemes were widely implemented in famous privacy-preserving blockchain systems, like Monero~\cite{Monero} and Bytecoin~\cite{Bytecoin}. Users can utilize RS schemes to obscure transactions by including chaff coins, called ``\textbf{mixins}", along with coins they will spend. As shown in Figure~\ref{fig:coins}, in UTXO model blockchains, there are multiple inputs and outputs in a transaction. Each input is represented by a RS, and each output is a coin. Each RS contains a coin which is spent in the transaction, marked in yellow color, and many other coins as mixins marked in green color. Each coin can only be spent in a RS. The diversity of a RS's mixins is measured by the number of historical transactions outputting them. Since the RS scheme's efficiency (generation and verification time) affects its verification time and transaction fee, researchers are motivated to improve its efficiency from $\mathbb{O}(n)$ in~\cite{van2013cryptonote} to $\mathbb{O}(\log n)$ in~\cite{yuen2019ringct}, where $n$ is the number of mixins in each RS. 

In contrast to extensive researches in the RS scheme's efficiency, the researches of its effectiveness on preserving privacy remain rather scarce. Currently, the effectiveness of a RS is roughly estimated as the number of mixins, like many k-anonymity methods~\cite{sweeney2002k}~\cite{li2007t}, and the current RS scheme randomly picks mixins~\cite{Monero}. 
The current RS scheme has two pivotal shortcomings: \textbf{a)} it is vulnerable to ``chain-reaction'' analysis~\cite{moser2018empirical}; \textbf{b)} it lacks of consideration on the diversity of mixins~\cite{chervinski2019floodxmr}. 

For the first weakness of the current RS scheme, through the ``chain-reaction'' analysis, it is possible to infer spent coins in RSs by leveraging the traffic flow and eliminating the mixins which must have been spent in other RSs. A RS $r$ can be the significant affected by the related RSs which have overlapped coins with $r$. We illustrate the ``chain-reaction'' analysis in Example \ref{ex:analysis}. For simplicity, in the remaining paper, we consider a RS as the union of mixins and the spent coin. \vspace{-1.5ex}

\begin{example}[``Chain-Reaction'' Analysis]
	\label{ex:analysis}
	There are four RSs, $r_1$ = $\{c_1, c_2, c_3\}$, $r_2$ = $\{c_1, c_2\}$, $r_3$ = $\{c_1, c_2\}$ and $r_4$ = $\{c_1, c_2, c_3, c_4$, $c_5, c_6, c_7\}$, where $c_1\sim c_7$ are 7 coins. Then, for $c_1$ and $c_2$, $r_2$ produces \textbf{observations} to witnesses with high similarity, since the probability of each coin being spent in $r_2$ is the same. The observation of a coin in a RS is the probability of the coin being spent in the RS, when the coin is spent. But for $c_2$ and $c_3$, $r_1$ produces observations with low similarity, since it is deduced that $c_3$ is the spent coin in $r_1$. Since $r_2$ and $r_3$ only contain $c_1$ and $c_2$, $c_1$ or $c_2$ is either spent in $r_2$ or $r_3$ (although detailed matches are unknown), and $c_3$ must be the spent coin in $r_1$.
	Besides, $r_4$ has 3 useless mixins (i.e., $c_1, c_2, c_3$), which increases $r_4$'s transaction fee. 
\end{example}\vspace{-1.5ex}

As the second weakness of the current RS scheme, it does not consider the impact of the diversity of mixins~\cite{chervinski2019floodxmr}. The diversity of a RS's mixins is measured by the number of historical transactions outputting them. When the diversity of mixins is low, the RS's effectiveness is low. In Example~\ref{ex:analysis}, if $c_5$ $\sim$ $c_7$ are outputted by the same historical transaction $t$ ($c_5$ $\sim$ $c_7$ were spent in $t$) and $c_4$ is the spent coin in $r_4$, then the owner of the historical transaction $t$, who is not the user that spends coin $c_4$, can deduce that $c_4$ is the spent coin in $r_4$ (as $c_1$ $\sim$ $c_3$ and $c_5\sim c_7$ are already detected).

To avoid the ``chain-reaction'' analysis, we need to consider the related RSs' impact and make RSs produce observations with ``similar'' distribution for any two coins. To overcome the second defect, we need to find a mixin set with high diversity. Since the transaction fee is proportional to the number of mixins, the users usually want to restrain the number of mixins within limited budget. Thus, we need to pick a set of mixins with high diversity and effectively resist ``chain-reaction'' analysis based attacks under the constraint of the budget.
To tackle this problem, two challenges need to be addressed: (1) for any two coins, how to measure the ``similarity level" of the observations; and (2) how to pick a desired set of mixins to maximize their diversity under the constraint of budget. 

In this paper, we propose a novel differential privacy concept, namely \underline{$\epsilon$}-\underline{c}oin-\underline{i}ndistinguishability ($\epsilon$-CI) to measure a RS's effectiveness. A RS's sender has $\epsilon$-privacy if any two coins in the RS can produce observations with ``similar" distributions, where the ``level of similarity" depends on $\epsilon$. The smaller $\epsilon$ is, the higher the privacy is. 
In the sequel, we illustrate the problem in a motivation example.\vspace{-1ex}

\begin{example}[The Coin-Indistinguishability-Aware Mixins Selection Problem]
	\label{ex:CIA-MS}
	Suppose there are four coins ($c_1, c_2, c_3, c_4$) and two RSs, $r_1 = \{c_1, \underline{c_2}\}$ and $r_2 = \{\underline{c_1}, c_3\}$. The spent coins in RSs are underlined. Among four coins, $c_1$ and $c_4$ are the same historical transaction's outputs while $c_2$ and $c_3$ are outputs of another two historical transactions. The budget is 3. Assume the required CI is very relaxed and only requires that the spent coin in a RS cannot be inferred. Now we want to generate a RS to spend $c_3$. As shown in Figure \ref{fig:motivation}, the permutation of possible spent coins under given RSs can be presented with a permutation tree, where the nodes in each level indicate the possible spent coins for the corresponding RSs.
	
	The first solution is using RS $\tau_1 = \{c_1, c_3\}$.  As shown in Fig. \ref{fig:solution1}, although $\tau_1$ cannot be inferred, it makes  witnesses  easily deduce that $c_2$ is the spent coin in $r_1$. For simplify, in the rest of this paper, we will only present the valid possible nodes of the permutation tree and ignore the corresponding RSs.
	
	The second solution is using RS $\tau_2 = \{c_1, c_3, c_4\}$. As shown in Fig. \ref{fig:solution2}, the CIs of $\tau_2$, $r_1$, and $r_2$ are preserved. However, its diversity is not large. Since $c_1$ and $c_4$ are the outputs of the same historical transaction, $\tau_3$'s diversity is only 2.
	
	The third solution is using RS $\tau_3 = \{c_1, c_2, c_3, c_4\}$. As shown in Fig. \ref{fig:solution3}, the CIs of $\tau_3$, $r_1$, and $r_2$ are preserved. Its diversity is large, which is 3. However, $\tau_3$ does not meet the budget constraint, since its cardinality is 4, which is larger than the budget.
	
	A good solution is using RS $\tau_4 = \{c_2, c_3, c_4\}$. As shown in Fig. \ref{fig:solution4}, the CIs of $\tau_4$, $r_1$, and $r_2$ are preserved. Besides, its diversity is large and it meets the budget constraint.
\end{example}\vspace{-1ex}

In this paper, we first formally define 
the \underline{c}oin-\underline{i}ndistinguishability-\underline{a}ware \underline{m}ixins \underline{s}election  with \underline{d}isjoint-\underline{s}uperset constraint (CIA-MS-DS) problem, which aims to find a mixin set that meets the required CI constraint, as well as the budget constraint, and has a maximal diversity. Moreover, each RS is either disjoint with or the superset of its preceding RSs. We prove the CIA-MS-DS problem is NP-hard through a reduction from the 01-knapsack problem~\cite{2013approximation} and we propose two approximation algorithms with theoretic guarantees. The Progressive Algorithm gradually narrows down the candidate mixin set one constraint by one constraint. The Game theoretic Algorithm models the problem as a game and let each mixin computes the right to be selected in the new RS.

To our knowledge, this is the first study to apply differential-privacy schemes on Blockchain systems and formally estimate a RS's effectiveness. Our work provides new insights into how to pick mixins in RS schemes. Specifically, we make the following contributions: \vspace{-1.5ex}
\begin{itemize}[leftmargin=*]
	\item We define the notion of $\epsilon$-coin-indistinguishability by applying differential-privacy schemes on Blockchain systems in Section～\ref{sec:CI}. \vspace{-4ex}
	\item We formally define the \underline{c}oin-\underline{i}ndistinguishability-\underline{a}ware \underline{m}ixins \underline{s}election with \underline{d}isjoint-\underline{s}uperset constraint (CIA-MS-DS) problem and give the proof of its hardness in Section \ref{sec:problemDefinition}.\vspace{-2ex}
	\item We propose two approximation algorithms, the Progressive Algorithm and the Game Theoretic Algorithm, with theoretic guarantees for the CIA-MS-DS problem in Section~\ref{sec:dp} and Section~\ref{sec:game} respectively.\vspace{-2ex}
	\item We conduct extensive experiments on both real and synthetic data sets and show efficiency as well as the effectiveness of our proposed solutions in Section \ref{sec:experimental}. \vspace{-0.5ex}
\end{itemize}

Besides, we propose a framework, CoinMagic, in Section~\ref{sec:framework}, discuss the related work in Section~\ref{sec:related} and conclude in Section~\ref{sec:conclusion}.
	\section{Coin-Indistinguishability}
\label{sec:CI}

In this section, we formalize the concept of coin-indistinguishability. As aforementioned, coin-indistinguishability is utilized to guarantee that it is hard to distinguish the spent coin in a RS. Our proposal is based on a generalization of differential privacy~\cite{dwork2011differential}. Our notion and technique abstract from the side information of the adversary, such as prior probabilistic knowledge about a RS's mixins. The advantages of the independence from the prior are that: first, the mechanism is designed for any prior probabilistic. Second, and even more important, it is also applicable when we do not have the information of the prior probabilistic of mixins~\cite{andres2012geo}. Because RSs can hide the spent coins, we even do not know whether the selected mixins have been spent or not.

\subsection{The Related RS Set and Mixin Universe}
As shown in Example~\ref{ex:CIA-MS}, a RS's effectiveness is impacted by some related RSs. We first formally define the mixin universe and the related RS set. Suppose $c_\tau$ is the coin that a user wants to spend in a new RS.

\begin{definition}
	(Mixin Universe) The mixin universe $\mathbb{C}_n = \{c_1, c_2,$ $\cdots, c_n\}$ is a set of coins that can be picked up as mixins in a new RS. A coin $c_i$ is the output of a transaction $t_i$.
\end{definition}

In blockchain systems, each coin is a transaction's output. In a transaction, there may be more than one output. In Example~\ref{ex:CIA-MS}, the mixin universe is $\mathbb{C} = \{c_1, c_2,$ $\cdots, c_4\}$.

\begin{definition}
	(Related Ring Signature Set) For a RS $r_f=\mathbb{C}_x\cup \{c_{\tau}\}$, the related RS set $\mathbb{R}_f = \{r_1$, $r_2$, $\cdots, r_m\}$ is a set of RSs with earlier spending timestamps than $r_f$ which contain the coin $c_\tau$ or any coins in $\mathbb{C}_x$.
\end{definition}

For instance, in Example~\ref{ex:CIA-MS}, the related RS set is $\{r_1, r_2\}$.
Since RSs in $\mathbb{R}_f$ may contain common coins with the new RS $r_f$, they may impact the effectiveness of $r_f$ on privacy reserving. For a set of RSs $\mathbb{R}_f$, let $I_f(r_i)$ indicate the position of RS $r_i$ in the ascendingly ordered list of RSs in $\mathbb{R}_f$ sorted according to their spending timestamps. 

\subsection{Probabilistic Model of MIXINS}

Since we do not know whether the mixins are spent or not, we introduce a probabilistic model here. Probabilities come into place in two ways. First, the adversary may have side information about the coins' expense, (e.g., knowing that some coins contained in a RS are not the spent coins since the adversary is the owner of these coins~\cite{chervinski2019floodxmr}). The adversary's side information can be modeled by a prior distribution $\pi(c,r)$ indicating the probability of coin $c$ being spent in  RS $r$.
Second, RS $r$ spending coin $c$ is also a probabilistic event. Since the spent coin is obscured by mixins, any coin in the RS $r$ is likely to be the spent coin.

Since RSs may not be disjoint and each coin can only be spent in a RS, given a set of RSs, there may be more than one possible \textit{spent coin permutation} over the given RS set. Here, we formally define a spent coin permutation as follow.

\begin{definition}[Spent Coin Permutation]
	\label{def:spent coin permutation}
Given a related RS set, $\mathbb{R}_f = \{r_1, r_2,$ $\cdots, r_m\}$, a spent coin permutation $\mathbb{P}$ is an ordered list of coins $[sc_1$, $sc_2,$ $\cdots,$ $sc_m]$, where $sc_i$ is a spent coin of a RS $r_j\in\mathbb{R}_f$ whose $I_f(j)$ is $i$ and $\forall sc_h, sc_g \in \mathbb{P}$, $sc_h \neq sc_g$.
\end{definition}

As shown in Fig. \ref{fig:motivation}, we represent these permutations by a tree structure for easier understanding. 

\begin{definition}[Spent Coin Permutation Tree]
	\label{def:spent coin permutation tree}
 Given a related RS set, $\mathbb{R}_f = \{r_1, r_2,$ $\cdots, r_m\}$, we can generate a spent coin permutation tree $\mathbb{T}_{f}$. Except for the root node, each node in $\mathbb{T}_{f}$ is presented by $\kappa=\langle d_\kappa, \mathbb{P}_\kappa, sc_\kappa \rangle$, where $d_\kappa$ is the node's depth in $\mathbb{T}_{f}$ (root's depth is 0), $\mathbb{P}_\kappa$ is a spent coin permutation $[sc^\kappa_1, sc^\kappa_2$, $\cdots, sc^\kappa_{d_\kappa-1}]$, where each $sc^\kappa_i$ is the spent coin of a RS $r_j\in\mathbb{R}_f$ whose $I_f(j) = i$, and $sc_\kappa$ is the spent coin in this node (i.e., $sc_\kappa=sc^\kappa_{d_\kappa}$). For node $\kappa$, $sc_\kappa \notin \mathbb{P}_\kappa$. For any two nodes, $\kappa$ and $\kappa'$, $\mathbb{P}_{\kappa} \cup sc_{\kappa} \neq \mathbb{P}_{\kappa'} \cup sc_{\kappa'}$, and if  $\mathbb{P}_{\kappa'} = \mathbb{P}_{\kappa} \cup sc_{\kappa}$, $\kappa$ is the father node of $\kappa'$.
\end{definition}

We can estimate the probability of each coin-and-signature pair by constructing its spent coin permutation tree. 
We denote $\mathbb{N}_{i,j}$ as the set of nodes in $\mathbb{T}_i$ whose depth is $I_i(j)$. Let $\mathbb{N}_{i,j}^{k,t}$ be the set of nodes in $\mathbb{T}_i$ whose depth is $I_i(j)$ and $k^{th}$ element in $\mathbb{P}_{\kappa}$ is coin $c_t$.

We can calculate the probability of coin $c_t$ being spent in the RS $r_k$ when the permutation tree is $\mathbb{T}_i$ by:\vspace{-2ex}
\begin{equation}
\label{eq:pr_f(c,r)}
Pr_i(c_t, r_k) = \frac{|\mathbb{N}_{i,i}^{k,t}|}{|\mathbb{N}_{i,i}|}
\end{equation}

Thus, the probability of $c_t$ having been spent in $\mathbb{R}_{i}$ can be calculated as:\vspace{-2ex}
\begin{equation}
\label{eq:pr_f(c)}
Pr_i(c_t) = \frac{\sum_{k=1}^{i}|\mathbb{N}_{i,i}^{k,t}|}{|\mathbb{N}_{i,i}|}
\end{equation} 

\begin{figure}[t]
	\subfigure[][{$\mathbb{T}_2$}]{
		\scalebox{0.27}[0.27]{\includegraphics{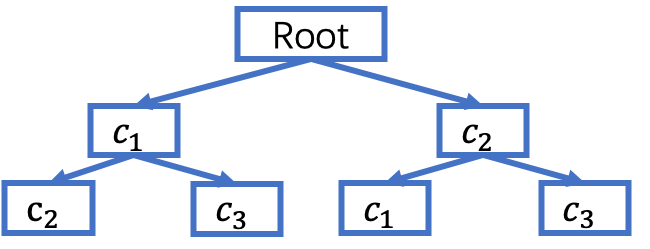}}
		\label{fig:pt_2}}
	\subfigure[][{$\mathbb{T}_3$}]{
		\scalebox{0.27}[0.27]{\includegraphics{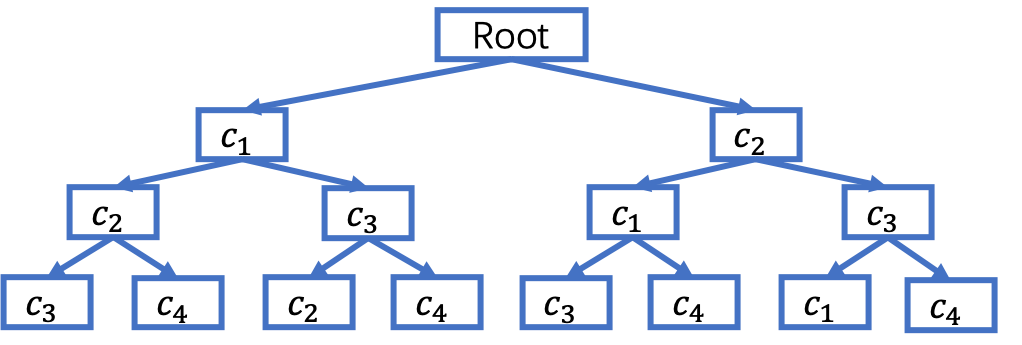}}
		\label{fig:pt_3}}
	\vspace{-3ex}
	\caption{Example \ref{ex:PT}.} \vspace{-4ex}
	\label{fig:PT}
\end{figure}

\begin{example}
	\label{ex:PT}
	There are three RSs, $r_1 = \{c_1, c_2\}$, $r_2 = \{c_1, c_2, c_3\}$, $r_3 = \{c_1, c_2, c_3, c_4\}$. The Figure \ref{fig:pt_2} shows the permutation tree $\mathbb{T}_2$ and Figure \ref{fig:pt_3} shows the permutation tree $\mathbb{T}_3$. Then, $|\mathbb{N}_{2,2}| = 4$, $|\mathbb{N}_{3,3}| = 8$, $|\mathbb{N}_{2,2}^{2,3}| = 2$, 
	$Pr_2(c_3, r_2) = \frac{1}{2}$, and $Pr_3(c_3, r_2) = \frac{1}{2}$.
\end{example}

\subsection{Coin-Indistinguishability}


We propose a formal definition of the coin-indistingushability to restrict the information leakage from observations. In other words, any adversary cannot obtain extra information from RSs w.r.t. the pre-defined privacy reserving level.

\begin{definition}
	\label{def:CI}
	(\underline{$\epsilon$}-\underline{C}oin-\underline{I}ndistinguishability ($\epsilon$-CI)) Given a related RS set, $\mathbb{R}_f = \{r_1, r_2, \cdots, r_m\}$, a RS $r_k \in \mathbb{R}_f$ satisfies $\underline{\epsilon}$-\underline{c}oin-\underline{i}ndistinguishability ($\epsilon$-CI) if for any two coins $c_i, c_j$ which are both contained in the RS $r_k$ (i.e., $c_i, c_j \in r_k$): 
	\begin{equation*}
	Pr_f(r_k|c_i)\leq e^{\epsilon}\cdot Pr_f(r_k|c_j)
	\end{equation*}
\end{definition}

$Pr_f(r_k|c_i)=\frac{Pr_f(r_k, c_i)}{Pr_f(c_i)}$ is the probability that when $c_i$ is spent, it is spent in $r_k$. For instance, in Example \ref{ex:PT}, $r_3$ satisfies $\frac{\ln 8}{3}$-CI. If $r_k$ satisfies $\epsilon$-CI, for any $\bar{\epsilon} \geq \epsilon$, $r_k$ also satisfy $\bar{\epsilon}$-CI. 
\subsection{Bounded Information Leakage}

We prove a characterization that quantifies over all priors to explain how the prior affects the privacy guarantees.

\begin{theorem}
	An upper bound of the posterior distribution of coin $c_i$ being spent in a RS $r_k \in \mathbb{R}_f$ which satisfies $\epsilon$-CI can be obtained by
	$Pr_f(c_i|r_k) \leq e^{\epsilon} \cdot \frac{\pi(c_i, r_k)}{\sum_{c_j \in r_k} \pi(c_j, r_k)}$, where $\pi(c_i, r_k)$ is the prior distribution modeled by the adversary's side information, indicating the probability of coin $c_i$ being spent in the RS $r_k$.
\end{theorem}

\begin{proof}
	We can calculate the posterior distribution of coin $c_i$ being spent in $r_k$ by 
	$Pr_f(c_i|r_k) = \frac{\pi(c_i, r_k) \cdot Pr_f(r_k|c_i)}{\sum_{c_j \in r_k} \pi(c_j, r_k) \cdot Pr_f(r_k|c_j)}. $
	Since $r_k$ satisfies $\epsilon$-CI, we have 
	$Pr_f(c_i|r_k)  = \frac{\pi(c_i, r_k) \cdot Pr_f(r_k|c_i)}{\sum_{c_j \in r_k} \pi(c_j, r_k) \cdot Pr_f(r_k|c_j)}$ $\leq \frac{\pi(c_i, r_k)}{e^{-\epsilon}\sum_{c_j \in r_k} \pi(c_j, r_k)} = e^{\epsilon} \cdot \frac{\pi(c_i, r_k)}{\sum_{c_j \in r_k} \pi(c_j, r_k)}$.
\end{proof}\vspace{-1ex}

The upper bound of posterior probability implies that no matter what prior information the adversary has, $\epsilon$-CI constrains the multiplicative distance between the posterior distribution and prior distribution within $e^\epsilon$, and thus limits the posterior information gain of the adversary.
\vspace{-2ex}
	\section{Problem Definition}
\label{sec:problemDefinition}

In this section, we first formally introduce some notions and then formulate the \underline{C}I-\underline{a}ware \underline{m}ixins \underline{s}election problem with \underline{d}isjoint-\underline{s}uperset constraint (CIA-MS-DS). Besides, we give the proof of properties of the CIA-MS-DS problem, which can help to solve the problem more efficiently. Furthermore, we give the proof of the NP-hardness of the CIA-MS-DS problem.\vspace{-1ex}

\subsection{Preliminaries}
\label{subsec:pre}

We first define the $\underline{\epsilon}$-CI-keeping RS as follow: \vspace{-1ex}
\begin{definition}
	\label{def:CI-keeping}
	(A $\underline{\epsilon}$-\underline{CI}-\underline{K}eeping \underline{R}ing \underline{S}ignature ($\epsilon$-CIK-RS)) Given a related RS set $\mathbb{R}_f$, a RS $r_{f}$ is a $\epsilon$-CI-keeping RS if $\forall i \leq f, \forall c, c' \in r_i, Pr_{f}(r_i|c) \leq Pr_{f}(r_i|c')\cdot e^\epsilon$.
\end{definition}\vspace{-1ex}

As aforementioned, subsequent RSs have impacts on the effectiveness of the previous RSs. To preserve the effectiveness existing RSs, a new RS should be a $\epsilon$-CIK-RS.

When each RS $r_i$ in a related RS set $\mathbb{R}_f$ is disjoint with any RS $r_j$ in $\mathbb{R}_i$ which is not the subset of $r_i$, the related RS set $\mathbb{R}_f$ is a \textbf{disjoint-superset related RS set}. It is practical for the real world applications since it does not divulge the users' privacy. Besides, it is helpful for calculating the conditional probabilities $Pr_f(r_k|c_i)$ for calculating the level of CI. We will introduce and give the proof os some properties of a disjoint-superset related RS set shortly in Subsection~\ref{subsec:properites}.\vspace{-1ex}

\begin{definition}
	\label{def:DS-RRSS}
	(Disjoint-Superset Related Ring Signature Set). For a related RS set, $\mathbb{R}_{f}$, it is a disjoint-superset related RS set if for any two RSs, they are disjoint or one RS is the superset of another RS.
\end{definition}\vspace{-1ex}

For instance, in Example~\ref{ex:PT} $\{r_1, r_2, r_3\}$ is a disjoint-superset related ring signature set.  However, in Example~\ref{ex:CIA-MS} $\{r_1, r_2\}$ is not a disjoint-superset related ring signature set, since $r_1 \cap r_2 = c_1$. 

By Definition~\ref{def:DS-RRSS}, given a disjoint-superset RS set $\mathbb{R}_f$, there are some special RSs in $\mathbb{R}_f$, namely \textit{super ring signature}, whose subsequent RSs are all disjoint with it. We formally define this kind of RSs as follows. \vspace{-1ex}

\begin{definition}
	\label{def:SRS}
	(Super Ring Signature) Given a related RS set $\mathbb{R}_{f}$, a RS $r_i \in \mathbb{R}_f$ is a super RS if for any $r_j \in \mathbb{R}_{f}$ whose $I_f(r_i) < I_f(r_j)$, $r_i \nsubseteq r_j$. 
\end{definition}\vspace{-1ex}

By Definition~\ref{def:DS-RRSS}, if a RS $r_i$ is a super RS, $\forall I_f(j) \geq I_f(i)$, $r_i \cap r_j =\emptyset$. For example, there are four RSs, $r_1 = \{c_1, c_2\}$, $r_2 = \{c_1, c_2, c_3\}$, $r_3 = \{c_1, c_2, c_3\}$, and $r_4 = \{c_4, c_5\}$. Assume $r_4$ is the RS whose timestamp is the latest and $I_4(r_1)<I_4(r_2)<I_4(r_3)<I_4(r_4)$. Then, $r_3$ and $r_4$ are super RSs. Since $I_4(r_1) < I_4(r_3)$ and $r_1\subset r_3$, $r_1$ is not a super RS. Similarly, $r_2$ is also not a super RS.

To more easily propose the problem and discuss its properties, we introduce some notations here. Given a related RS set $\mathbb{R}_f$, we denote $\mathbb{SRS}_f$ as the set of super RS, $\mathbb{SRS}_f = \{srs_1, srs_2$, $\cdots, srs_n\}$. For each super RS, $srs_i$, we define its diversity, $dive_i$, as the number of historical transactions outputting the coins in $srs_i$. We denote $ns_i$ as the number of RSs in the related RS set $\mathbb{R}_f$ which are subsets of $srs_i$. We denote the degree $d_i$ of $srs_i$ as the number of coins in $srs_i$ minus $ns_i$ (i.e., $d_i = |srs_i| - ns_i$). Besides, we define the maximal coin spent probability, $pr_{max}^i$, of $srs_i$ as the maximal value of the probability that a coin in $srs_i$ has been spent in the related RS set $\mathbb{R}_f$, i.e., $pr_{max}^i = \max \{Pr_f(c)|c \in srs_i\}$, where $Pr_f(c)$ is defined as the probability that the coin $c$ having been spent in $\mathbb{R}_f$ in Equation~\ref{eq:pr_f(c)}. Similarly, we define the minimal coin spent probability, $pr_{min}^i$, of $srs_i$ as the minimal value of the probability that a coin in $srs_i$ has been spent in the related RS set $\mathbb{R}_f$, i.e., $pr_{max}^i = \min \{Pr_f(c)|c \in srs_i\}$.

Among the mixin universe, there may be some coins that are not contained in any RSs. In Example~\ref{ex:CIA-MS}, when we try to generate a RS to spend $c_3$, $c_4$ has not been contained in any RSs.\vspace{-1ex}

\begin{definition}
	\label{def:fresh coins}
	(Fresh Coin Set) Given a related RS, $\mathbb{R}$, and a mixin universe $\mathbb{C}$, a fresh coin set $\mathbb{F} = \{fc_1, fc_2, \cdots, fc_n \}$ is a set of coins in $\mathbb{C}$ that have not been contained in any RSs in $\mathbb{R}$.
\end{definition}




\vspace{-3ex}
\subsection{The CI-aware Mixins Selection with Disjoint-superset Constraint Problem}

In this subsection, we formally define the CIA-MS-DS problem. \vspace{-3ex}

\begin{definition}
	(The \underline{CI}-\underline{a}ware \underline{m}ixins \underline{s}election with \underline{d}isjoint-\underline{s}uperset constraint (CIA-MS-DS) problem) Given a super RS set $\mathbb{SRS}$, a fresh coin set $\mathbb{F}$, the coin $c_\tau$ that will be spent, a required $\epsilon$, and a budget $B$, a user wants to pick up a set of mixins, combining $c_\tau$, to generate the new RS $r_{\tau}$, such that its diversity $dive_{r_{\tau}} = |\{t_i | c_i \in r_{\tau}\}|$ is maximized and the following constraints are satisfied:\vspace{-1.5ex}
	\begin{itemize}[leftmargin=*]
		\item \textbf{DS constraint} $r_{\tau}$ is composed of some super RSs in $\mathbb{SRS}$ and some fresh coins in $\mathbb{F}$;\vspace{-1.5ex}
		\item \textbf{Budget constraint} the number of coins in $r_{\tau}$ does not exceed the budget; and\vspace{-1.5ex}
		\item \textbf{$\epsilon$-CIK constraint} $r_{\tau}$ is a $\epsilon$-CIK-RS.\vspace{-1.5ex}
	\end{itemize}\label{def:problem}
\end{definition}\vspace{-2ex}

Since the transaction fee of a RS is proportional to the number of coins in it, to limit the transaction fee of the new RS $r_\tau$, the number of coins in $r_\tau$ should meet the budget constraint, i.e., $| \{c| c\in r_{\tau}\}| \leq B$. Besides, since the new RS $r_\tau$ should protect effectiveness of the existing RSs, it should be a $\epsilon$-CIK-RS. Besides, when the new RS $r_\tau$ is composed of some super RSs in $\mathbb{SRS}$ and some fresh coins in $\mathbb{F}$, $r_\tau$ is disjoint with any RS in $\mathbb{R}_f$ which is not the subset of $r_\tau$. This can help to quickly verify if the new RS $r_\tau$ is a $\epsilon$-CIK-RS by the properties of the CIA-MS-DS problem, which will be introduced shortly in the next subsection. In addition, it can help other users to quickly generate the new RSs, which makes the problem practical for real world applications. \vspace{-1ex}

\subsection{Properties of the CIA-MS-DS Problem}
\label{subsec:properites}
When a related RS set is a disjoint-superset related RS set, there are some important properties, which can help to calculate the attributes of each super RS (Theorem ~\ref{theorem:pr(r,a)01},~\ref{theorem:pr_m(r_m,a)01} and~\ref{theorem:pr(a)01}) and reduce the time complexity of verifying whether the generated RS $r_\tau$ is a $\epsilon$-CIK-RS (Theorem~\ref{theorem:pr(a)dandiao},~\ref{theorem:Pr(|)daoPr(c)} and ~\ref{theorem:epsin}).

We first prove in Theorem~\ref{theorem:pr(r,a)01} that when a related RS set is a disjoint-superset related RS set, the probability of coin $c$ being spent in the RS $r$ will keep stable. \vspace{-1ex}

\begin{theorem}
	\label{theorem:pr(r,a)01}
	Given a disjoint-superset related RS set, $\mathbb{R}_f = \{r_1, r_2$, $\cdots, r_m\}$, $\forall I_f(i)\in [2, m], \forall I_f(j) < I_f(i)$, $\forall c_t \in r_j$, we have $Pr_{i}(r_j, c_t)$ =$ Pr_{h}(r_j, c_t)$, where $I_f(h) = I_f(i) -1$.
\end{theorem}\vspace{-3ex}

\begin{proof}
	By Definition~\ref{def:spent coin permutation tree} and ~\ref{def:DS-RRSS}, $\forall i \in [1,m]$, for each node $\kappa$ in $\mathbb{N}_{i, i}$, we can partition $\mathbb{P}_\kappa$ as two sets $\mathbb{P}_\kappa^1$ and $\mathbb{P}_\kappa^2$, where $\mathbb{P}_\kappa^1 \subseteq r_i$, $\mathbb{P}_\kappa^2 \cap r_i = \emptyset$, and $\mathbb{P}_\kappa^1 \cup \mathbb{P}_\kappa^2 = \mathbb{P}_\kappa$. Thus, $|\mathbb{P}_\kappa^1| = ns_i$. Since each RS spents an unspent coin, $\forall i \in [1, m]$, $|r_i| > ns_i$. Thus, by Definition~\ref{def:spent coin permutation tree}, $\forall I_f(i) \in [2,m]$, $\mathbb{N}_{i,i} = \mathbb{N}_{i,h}$ and for each node $\kappa \in \mathbb{N}_{i, h}$, it has $d_i$ children. Thus, $\forall I_f(i) \in [2,m], \forall I_f(j)<I_f(i), \forall c_t \in r_j, \mathbb{N}_{i,h}^{j, t} = \mathbb{N}_{h,h}^{j,t}$.Thus, by Equation~\ref{eq:pr_f(c,r)}, $\forall I_f(i) \in [2,m]$, $\forall I_f(j)$ $< I_f(i), \forall c \in r_j$, $Pr_i(r_j, c_t) = \frac{|\mathbb{N}_{i,i}^{j,t}|}{|\mathbb{N}_{i,i}|} = \frac{|\mathbb{N}_{i,h}^{j,t}| \cdot d_i}{|\mathbb{N}_{i,h}|\cdot d_i}= $ $ \frac{|\mathbb{N}_{i,h}^{j,t}|}{|\mathbb{N}_{i,h}|} =Pr_{h}(r_j, c_t)$.
\end{proof}\vspace{-2ex}

Next, we prove in Theorem~\ref{theorem:pr_m(r_m,a)01} that when a related RS set is a disjoint-superset related RS set, the probability of a coin $c_t$ being spent in the latest RS can be calculated iteratively. \vspace{-1ex}

\begin{theorem}
	\label{theorem:pr_m(r_m,a)01}
	Given a disjoint-superset related RS set $\mathbb{R}_f = \{r_1, r_2$, $\cdots, r_m\}$, $\forall c_t \in r_f$, we have $Pr_f(r_f, c_t) = \frac{1 - Pr_{h}(c_t)}{d_f}$, where $I_f(h) = I_f(f)-1$. 
\end{theorem}\vspace{-3ex}

\begin{proof}
	As proved in Theorem~\ref{theorem:pr(r,a)01}, each node in $\mathbb{N}_{f, h}$ has $d_{f}$ children and $\forall c_k \in r_{h}, \mathbb{N}_{f,h}^{h, k} = \mathbb{N}_{h,h}^{h,k}$. Therefore, according to Equation~\ref{eq:pr_f(c,r)}, $\forall c_t \in r_f$, 
	we have $Pr_f(r_f, c_t) = \frac{|\mathbb{N}_{f,f}^{f,t}|}{|\mathbb{N}_{f,f}|} 
	=\frac{|\mathbb{N}_{f,h}|-\sum_{j=1}^{h}|\mathbb{N}_{f,h}^{j,t}|}{|\mathbb{N}_{f,h}|\cdot d_f} \\ 
	=\frac{\frac{|\mathbb{N}_{f,h}|}{|\mathbb{N}_{f,h}|}-\frac{\sum_{j=1}^{h}|\mathbb{N}_{f,h}^{j,t}|}{|\mathbb{N}_{f,h}|}}{d_f} = \frac{1- Pr_{h}(c_t)}{d_f}$.
\end{proof}\vspace{-2ex}

Then, we show 
that when a related RS set is a disjoint-superset related RS set, the probability of a coin $c$ having been spent in $\mathbb{R}_f$ can be calculated iteratively.\vspace{-1ex}

\begin{theorem}
	\label{theorem:pr(a)01}
	Given a disjoint-superset related RS set, {\scriptsize$\mathbb{R}_f = \{r_1, r_2$, $ \cdots, r_m\}$},  {\scriptsize$\forall i \in [1, m], \forall c$, $
		Pr_i(c) =
		\begin{cases}
		Pr_{j}(c) & \mbox{$c \notin r_i$} \\
		Pr_{j}(c) + \frac{1 - Pr_{j}(c)}{d_i} & \mbox{$c \in r_i$} 
		\end{cases},
		$}
	where $I_f(j) = I_f(i)-1$, $I_f(g) = 0$, and $Pr_g(c) = 0$. 
\end{theorem}\vspace{-3ex}

\begin{proof}
	As proved, $\forall i \in [2,m], \forall h < i, \forall c_t \in r_h$, $\mathbb{N}_{i,j}^{h, t} = \mathbb{N}_{j,j}^{h,t}$ and each node in $\mathbb{N}_{i, j}$ has $d_i$ children. Thus, $\forall i \in [2,m], \forall c_t \in r_i$, 
	$Pr_i(c_t)  = \frac{\sum_{h=1}^{i}|\mathbb{N}_{i,i}^{h,t}|}{|\mathbb{N}_{i,i}|}  =\frac{(|\mathbb{N}_{i,j}|-\sum_{h=1}^{j}|\mathbb{N}_{i,j}^{h,t}|)+d_i\cdot\sum_{h=1}^{j}|\mathbb{N}_{i,j}^{h,t}|}{|\mathbb{N}_{i,j}|\cdot d_i}$
	$= \frac{\sum_{h=1}^{j}|\mathbb{N}_{i,j}^{h,t}|}{|\mathbb{N}_{i,j}|}+ \frac{1 - \frac{\sum_{h=1}^{j}|\mathbb{N}_{i,j}^{h,t}|}{|\mathbb{N}_{i,j}|}}{d_i}
	= Pr_{j}(c_t) + \frac{1-Pr_{j}(c_t)}{d_i}$. 
	Let $I_f(k) = 1$. By Equation~\ref{eq:pr_f(c)}, {\scriptsize$\forall c \in r_k, Pr_k(c)$ = $\frac{1}{|r_k|}$ = $Pr_g(c)+\frac{1-Pr_g(c)}{d_k}$}. Similarly, we can prove $\forall i \in [1,m], \forall c \notin r_i, Pr_i(c) = Pr_{j}(c)$. 
\end{proof}\vspace{-2ex}

Thus, given a disjoint-superset related RS set, $\mathbb{R}_f$, we can easily calculate the attributes of each super RS. Then, for any two coins $c, c'$ in the same RS $r_i$, we proved that if the probability of $c$ being spent is smaller than the probability of $c'$ being spent in $\mathbb{R}_i$, the probability of $c$ being spent must be smaller than the probability of $c'$ being spent in any $\mathbb{R}_j$, where $I_f(j) \geq I_f(i)$.\vspace{-1ex}

\begin{theorem}
	\label{theorem:pr(a)dandiao}
	Given a disjoint-superset related RS set, $\mathbb{R}_f = \{r_1, r_2$, $\cdots, r_m\}$, $\forall i \in [1, m], \forall I_f(j) \leq I_f(i), \forall c, c' \in r_j$, if $Pr_j(c) \leq Pr_j(c')$, it holds that $Pr_{i}(c) \leq Pr_{i}(c')$.
\end{theorem}\vspace{-3ex}

\begin{proof}
	By Definition~\ref{def:DS-RRSS}, $\forall i \in [1,m], \forall I_f(j) \leq I_f(i), \forall c, c' \in r_j$, if $c \in r_i$, $c' \in r_i$. Suppose $I_f(j) = I_f(i)-1$. Thus,
	$
		Pr_{i}(c) - Pr_{i}(c') 
		= Pr_{j}(c)+\frac{1-Pr_j(c)}{d_i} - Pr_{j}(c') - \frac{1-Pr_j(c')}{d_i}$
	$=  [Pr_{j}(c)- Pr_j(c')] - \frac{1}{d_i} \cdot [Pr_j(c) - Pr_j(c')].
	$
	Since $d_i \geq 1$ and $Pr_j(c) \leq Pr_j(c')$, $Pr_{i}(c) - Pr_{i}(c') \leq 0$.
\end{proof}\vspace{-2ex}

Besides, given a disjoint-superset related RS set $\mathbb{R}_f$, for any two coins $c, c'$ in the same RS $r_j$, we proved that if the conditional probability $Pr_i(r_j|c)$ is smaller than $Pr_i(r_j|c')$, we can find for any $\forall I_f(k) \in[I_f(j)-1,I_f(i)]$, the probability of $c$ being spent is larger than the probability of $c'$ being spent. \vspace{-1ex}

\begin{theorem}
	\label{theorem:Pr(|)daoPr(c)}
	Given a disjoint-superset related RS set, $\mathbb{R}_f= \{r_1, r_2$, $\cdots, r_m\}$, $\forall i \in [1, m], \forall I_f(j) \leq I_f(i), \forall c, c' \in r_j$, if $Pr_i(r_j|c) \leq Pr_i(r_j|c')$, $\forall I_f(k) \in [I_f(g),I_f(i)], Pr_{k}(c) \geq Pr_{k}(c')$, where $I_f(g) = I_f(j)-1$. 
\end{theorem}\vspace{-3ex}

\begin{proof}
	Denote $F(t,j,i)$ as a function to recursively calculate $Pr_i(c_t)$ from $Pr_j(c_t)$. As proved in Theorem~\ref{theorem:pr(a)dandiao}, $F(t,j,i)$ increases when $Pr_{j}(c_t)$ increases. 
	By Definition~\ref{def:DS-RRSS}, $\forall i \in [1,m], \forall I_f(j) \leq I_f(h)$, $\forall c, c' \in r_j$,  if $c \in r_i$, $c' \in r_i$. 
	By Theorem~\ref{theorem:pr(r,a)01} and ~\ref{theorem:pr(a)01}, $\forall i \in [1, m]$, $\forall I_f(j)\leq I_f(i)$, $\forall c_t \in r_j$, $Pr_i(r_j|c_t)$ = $\frac{Pr_i(r_j, c_t)}{Pr_i(c_t)}$ = $\frac{Pr_j(r_j,c_t)}{F(t, j, i)}$ = $\frac{\frac{1-Pr_{g}(c_t)}{d_j}}{F(t, g, i)}$. Thus, $Pr_i(r_j|c_t)$ increases when $Pr_{g}(c_t)$ decreases.
	Thus, $\forall i \in [1, m]$, $\forall I_f(j) \leq I_f(i), \forall c, c' \in r_j$, if $Pr_i(r_j|c) \leq Pr_i(r_j|c')$, $Pr_{g}(c) \geq Pr_{g}(c')$. Then by Theorem~\ref{theorem:pr(a)dandiao}, we have $\forall i \in [1, m], \forall I_f(j) \leq I_f(i), \forall c, c' \in r_j$, $\forall I_f(k) \in [I_f(g), I_f(i)], Pr_{k}(c) \geq Pr_{k}(c')$.
\end{proof}\vspace{-2ex}

Then we prove that, if the conditional probabilities $Pr_{i}(r_j|c_t)$ of two coins in a RS satisfy $\epsilon$-CI, the conditional probabilities of these two coins in the RS after adding some RSs also satisfy $\epsilon$-CI.\vspace{-1ex}

\begin{theorem}
	\label{theorem:epsin}
	Given a disjoint-superset related RS set, $\mathbb{R}_f = \{r_1, r_2$, $\cdots, r_m\}$, $\forall i \in [1, m], \forall I_f(j) \leq I_f(i), \forall c, c' \in r_j$, if $Pr_i(r_j|c) \geq Pr_i(r_j|c')$ and $\frac{Pr_i(r_j|c)}{Pr_i(r_j|c')} = \beta$, $\forall k \in [I_f(i), m]$, $Pr_{k}(r_j|c) \geq Pr_{k}(r_j|c')$ and $\frac{Pr_k(r_j|c)}{Pr_k(r_j|c')} = \beta'\leq\beta$. 
\end{theorem}\vspace{-3ex}
\begin{proof}
	By Theorem~\ref{theorem:Pr(|)daoPr(c)}, if $Pr_i(r_j|c) \geq Pr_i(r_j|c')$, $Pr_i(c) \leq Pr_i(c')$. Suppose $r_k$ is the RS with the lowest $I_f(k)$ where $I_f(k) \geq I_f(i)$ and $r_i \subseteq r_k$. By Theorem~\ref{theorem:pr(a)dandiao}, $Pr_k(c) \leq Pr_i(c')$ and $Pr_{k}(r_j|c) \geq Pr_{k}(r_j|c')$. 
	Then, by Theorem~\ref{theorem:pr(r,a)01} and ~\ref{theorem:pr(a)01}, $ 
	\beta' = \frac{Pr_{k}(r_j|c)}{Pr_{k}(r_j|c')}
	= \frac{Pr_i(r_j,c)}{Pr_i(c) + \frac{1-Pr_i(c)}{d_i}} \cdot \frac{Pr_i(c') + \frac{1-Pr_i(c')}{d_i}}{Pr_i(r_j,c')}
	=  \frac{Pr_i(r_j,c)}{Pr_i(r_j,c')} \cdot \frac{(d_i-1)Pr_i(c')+1}{(d_i-1)Pr_i(c)+1}$. 
	When $d_i-1 \geq 1$, $\beta' \leq \frac{Pr_i(r_j,c)}{Pr_i(r_j,c')} \cdot \frac{Pr_i(c')}{Pr_i(c)} = \beta$. When $d_i-1 =0$, 
	$\beta' = \frac{Pr_i(r_j,c)}{Pr_i(r_j,c')} \leq \frac{Pr_i(r_j,c)}{Pr_i(r_j,c')} \cdot \frac{Pr_i(c')}{Pr_i(c)} = \beta$. Similarly, we can prove $\forall k \in [I_f(i), m] \land r_i \nsubseteq r_k$, it still holds.
\end{proof}\vspace{-2ex}

Suppose the corresponding disjoint-superset related RS set of the given super RS set $\mathbb{SRS}$ is $\mathbb{R}_\tau$. Thus, if the new RS $r_\tau$ is composed of some super RSs in $\mathbb{SRS}$ and some fresh coins in $\mathbb{F}$, by Definition~\ref{def:SRS}, $r_\tau$ is a super RS in $\mathbb{R}_\tau \cup r_\tau$. Then by Theorem~\ref{theorem:epsin}, since RSs in $\mathbb{SRS}$ all satisfy the $\epsilon$-CI, if $r_\tau$ satisfies $\epsilon$-CI, it is a $\epsilon$-CIK-RS. In other words, to verify if $r_\tau$ is a $\epsilon$-CIK-RS, we just need to verify whether it satisfies $\epsilon$-CI, which reduces the time complexity form $\mathcal{O}(\sum_{i=1}^{|\mathbb{R}_\tau|+1}|r_i|^2)$ to $\mathcal{O}(|r_{\tau}|^2)$. Besides, suppose $pr_{max} = \max \{pr_m(c)| c \in r_{\tau}\}$ and $pr_{min} = \min \{pr_m(c)| c \in r_{\tau} \}$. Thus, by Definition~\ref{def:CI}, Theorem~\ref{theorem:pr(a)01} and Theorem~\ref{theorem:pr_m(r_m,a)01}, to check $r_{\tau}$ satisfies $\epsilon$-CI, we just need to verify whether $\frac{(1-pr_{max})\cdot e^\epsilon}{(d_{\tau}-1)\cdot pr_{max}+1} \geq \frac{(1-pr_{min})}{(d_{\tau}-1)\cdot pr_{min}+1}$, which further reduces the time complexity from $\mathcal{O}(|r_{\tau}|^2)$ to $\mathcal{O}(|r_{\tau}|)$, where $d_\tau$ is the degree of $r_\tau$ and is calculated as the number of coins in $r_\tau$ minus the number of its subsets in $\mathbb{R}_\tau$. \vspace{-3ex}


\subsection{Hardness of the CIA-MS-DS Problem}
In this subsection, we prove that the CIA-MS-DS problem is NP-hard by reducing 0-1 knapsack problem~\cite{2013approximation}.\vspace{-1ex}

\begin{theorem}
	\label{theorem:np}
	The CIA-MS-DS problem is NP-hard.
\end{theorem}\vspace{-4ex}

\begin{proof}
	We prove that the theorem by a reduction from the 0-1 knapsack problem~\cite{2013approximation}: given a set, $I$, of $n$ items $i$ numbered from 1 up to $n$, each with a weight $w_i$ and a value $x_i$, along with a maximum weight capacity $C$, the 0-1 knapsack problem is to find a subset $I'$ of $I$ that maximizes $\sum_{i \in I'}x_i$ subjected to $\sum_{c_j \in r_k} w_i \leq C$.
	
	For any given 0-1 knapsack problem instance, we can transform it into a special CIA-MS-DS problem instance in polynomial time as follows: we generate a super RS set, $\mathbb{SRS}$, where there are $n$ super RS numbered from 1 up to $n$, each with a size $x_i$ and diversity $w_i$. The fresh coin set is empty. The $\epsilon$ is large enough that any combination of super RSs in $\mathbb{SRS}$ is a $\epsilon$-CIK-RS. Besides, for any two coins, the historical transactions which outputs them are different. Thus, $dive_{r_{\tau}} = \sum_{srs_i \in r_{\tau}} dive_{i}$.
	
	To find the new RS $r_{\tau}$ whose diversity is maximum is equal to find a maximum assignment of 0-1 knapsack problem. Given this mapping, we can reduce the 0-1 knapsack problem to the CIA-MS-DS problem. Since the 0-1 knapsack problem is NP-hard~\cite{2013approximation}, the CIA-MS-DS problem is at least NP-hard.
\end{proof}

	\section{CoinMagic Framework}
\label{sec:framework}

\begin{figure}[t]
	\centering
	\includegraphics[scale=0.5]{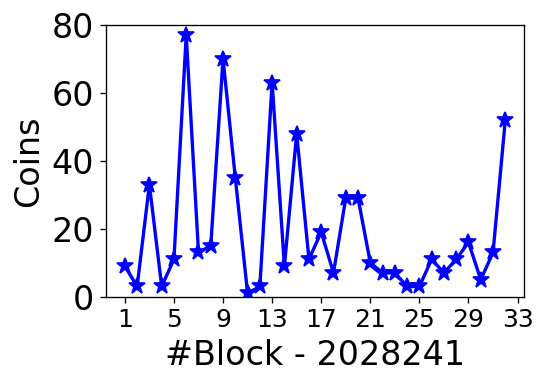}\vspace{-4ex}
	\caption{\small The Number of Coins in blocks in Monero blockchain}
	\label{fig:blocks}\vspace{-3ex}
\end{figure}

Since the new RS's effectiveness on privacy preserving is impacted by the related RSs, before selecting mixins, a user needs to retrieve the mixin universe $\mathbb{C}$ and the related RS set $\mathbb{R}$. The mixin universe can be retrieved according to the user's interest or the blockchain system's requirement. For instance, in Monero~\cite{Monero}, the system requires that in each RS, half of the mixins should be the coins which are less than 1.8 days old~\cite{hinteregger2019short}. However, this method has two defects. Firstly, since the number of transactions in each block is unstable (especially, some miners even generate empty blocks), the cardinality of $\mathbb{C}$ is also not constant. For instance, Figure~\ref{fig:blocks} shows the number of coins in blocks between 2028242 and 2028273 of Monero. Among these blocks, the block 2028252 only contains 1 coin but the block 2028247 contains 77 coins. Secondly, the size of the related RS set can unlimitedly increased over time, which is not efficient for constructing the spent coin permutation tree and generating new RSs.

In this paper, we propose a framework,  namely CoinMagic, to retrieve the mixin universe $\mathbb{C}$ and related RS set $\mathbb{R}$, then generate new RSs. Specifically, CoinMagic partitions the blocks in a blockchain into disjoint and sequential batches, then generates new RSs only upon the coins and RSs in the same batch. The number of coins in each batch is bounded in  range specified by the system. 
For each batch, its mixin universe is consisted of coins in the batch,  Thus, the mixin universes for different batches are disjoint. In addition, as each RS selects mixins from the mixin universe in its batch, the related RS sets in different batches are also disjoint. The size of each related RS set is bounded by the  cardinality of the mixin universe for its batch. 

As shown in Algorithm~\ref{al:framework}, mixin universes are retrieved by batches. We start at the last block of the last batch before $c_\tau$ (line 1). Then, we initialize the mixin universe $\mathbb{C}$ and the related RS set $\mathbb{R}$ as empty (line 1). Then we continually read blocks and add the outputs of transactions in blocks to the mixin universe $\mathbb{C}$ until we reach the last block in the blockchain or the cardinality of $\mathbb{C}$ is large enough (line 2-5). Then, we continually read blocks and add the RSs in blocks which are the subsets of $\mathbb{C}$ to $\mathbb{R}$ (line 6-10). Since the size of every block is limited, the number of coins in a block has an upper bound $\lambda'$. Thus, $|\mathbb{C}| \in [\lambda, \lambda + \lambda'-1]$ and $|\mathbb{R}| \leq |\mathbb{C}| \leq \lambda+\lambda'-1$. This procedure (line 1 -10) can be accomplished when the user updates her/his local blockchain status to the global status. In other words, the user does not need to pay much extra cost for implementing this framework. 
Finally, we run a RS generation algorithm to generate a new RS satisfying the constraints of DS, budget, and $\epsilon$-CIK.

Since the mixins of a RS can only be selected within the same batch,  there could be a situation that a user cannot find a RS with $\epsilon$-CI to spend a coin. For example, in Example~\ref{ex:CIA-MS}, if the user uses the RS $\tau_5 = \{c_1, c_2, c_3\}$, then when another user wants to spend $c_4$, she/he cannot find a RS $r$ with $\epsilon$-CI, since witnesses can easily find that $c_1, c_2,$ and $c_3$ has been spent in previous three RSs (i.e., $r_1$, $r_2$, $r_3$) and $c_4$ must be the spent coin in the new RS $r$. We denote the set of fresh coins in a batch as $\mathbb{F}$. To avoid this issue, we require that at any moment, the $\mathbb{F}$'s cardinality cannot be 1. Next, we formally prove that if $|\mathbb{F}| \neq 1$, for any $c_\tau$ and budget, there always exists at least one RS that satisfies the constraints of DS, budget and $\epsilon$-CIK.

\begin{algorithm}[t]
	\DontPrintSemicolon
	\KwIn{the spent coin $c_\tau$, blocks in the blockchain system, $\mathbb{B}$, the threshold, $\lambda$, of the number of coins in a batch, budget, and the system's required CI level}
	\KwOut{the new RS}
	set $b_i$ as the last block of the last batch before $c_\tau$, $\mathbb{C} = \emptyset$, $\mathbb{R} = \emptyset$;\;
	\While{$b_i$ is not the last block in $\mathbb{B}$ and $|\mathbb{C}| < \lambda$}{
		$b_i \gets$ next block;\;
		\If{$b$ is not an empty block}{
			add the output coins of transactions in block $b$ to $\mathbb{C}$;\;
		}
	}
	\While{$b_i$ is not the last block in $\mathbb{B}$}{
		$b_i \gets$ next block;\;
		\ForEach{RS $r$ in $b_i$}{
			\If{$r$ is a subset of $\mathbb{C}$}{
				add $r$ to $\mathbb{R}$;\;
			}
		}
	}
	run a RS generation algorithm to generate a new RS satisfying the constraints of DS, budget and $\epsilon$-CIK in Definition \ref{def:problem}.
	\caption{CoinMagic Framework}
	\label{al:framework}
\end{algorithm} 

\vspace{-1ex}
\begin{theorem}
	If $|\mathbb{F}| \neq 1$, for any $c_\tau$ and a budget $B$, there always exists at least one RS which satisfies the constraints of DS, $\epsilon$-CIK, and budget, where $m_\tau$ is the module which contains $c_\tau$ and $B \geq \max\{|m_\tau|, 2\}$.
\end{theorem}\vspace{-3ex}
\begin{proof}
	When $m_\tau$ is a super RS, according to Theorem~\ref{theorem:epsin}, $r_{m+1} = m_\tau$ is a $\epsilon$-CIK-RS and its size satisfies the budget constraint. When $m_\tau$ is a fresh coin, since $|\mathbb{F}| \neq 1$, there is at least one fresh coin, denoted by $fc$. By Definition~\ref{def:CI}, $r_{m+1} = m_\tau \cup fc$ is also a CIK-RS and its size satisfies the budget constraint.
\end{proof}

\vspace{-2ex}

	\section{Progressive Approach}
\label{sec:dp}

In this section, to tackle the CIA-MS-DS problem, we propose an approach, namely the Progressive Algorithm. The main idea of the approach is that we first use a 0-1 knapsack algorithm to find a selection which satisfies the $\epsilon$-CIK constraint and the DS constraint. Then we greedily change the selection to make it satisfy the budget constraint.

\subsection{Progressive Algorithm}

\begin{algorithm}[t]
	\DontPrintSemicolon
	{
	\KwIn{A spent coin $c_\tau$, a budget $B$, a super RS set $\mathbb{SRS}$, and a fresh coin set $\mathbb{F}$}
	\KwOut{An eligible RS $r_{\tau}$}
	$\mathbb{M} = \mathbb{SRS} \cup \mathbb{F}$, $r_{\tau} = m_\tau$;\;
	\For{$i = 1$ to $|M|$}{
		$M_{i} = M \backslash \bigcup_{m_k \in M \land Pr_{max}^{k} > Pr_{max}^{i}} m_k$;\;
		\For{$j = 1$ to $|M|$}{
			\If{{\scriptsize$m_j \in M_{i}$, $Pr_{max}^{i}\geq Pr_{max}^{\tau}$ and $Pr_{min}^{j} \leq Pr_{min}^{\tau}$}}{
				$\omega_{i,j} = m_i \cup m_j \cup m_\tau$;\;
				{\scriptsize$M_{i,j} = M_i  \backslash (\omega_{i,j} \cup\bigcup_{m_k \in M_i \land Pr_{min}^{k} < Pr_{min}^{j}}m_k)$};\;
				update the diversity of each $m$ in $M_{i,j}$;\;
				calculate $\bar{d}_{i,j}$ and $O_{max}^{i,j}$;\;
				$C_{i,j} =  \delta$-KP($M_{i,j}, \bar{d}_{i,j}, O_{max}^{i,j}$);\;
				\If{$|C_{i,j}| + |\omega_{i,j}| > B$}{
					$\psi_{i,j} = \omega_{i,j}$;\; 
					\While{$|\psi_{i,j}| \leq B$}{
						greedily add $\bar{m}$ in $M_{i,j}$ whose ratio of the increase of $\psi_{i,j}$'s diveristy to its size is the largest to $\psi_{i,j}$;\;
						$M_{i,j} = M_{i,j} \backslash \bar{m}$;\;
					}
				}\Else{$\psi_{i,j} = \omega_{i,j} \cup C_{i,j}$;\;}
			}
		}
	}
	\Return{the $\psi_{i,j}$ with maximal diversity as $r_{\tau}$}\;}
	\caption{Progressive Algorithm}
	\label{algo:kp}
\end{algorithm}

The algorithm is inspired by the following three lemmas. The first lemma shows that the degree of the new RS $r_\tau$ in $\mathbb{R}_f$, denoted by $d_\tau$, is the sum of the degree of each super RS and fresh coin which is contained in $r_\tau$, i.e., $d_\tau = \sum_{srs_i \subseteq r_\tau}d_i + \sum_{fc_i \subseteq r_\tau}d_i$. Similar with the definitions of a super RS's attributes in Subsection~\ref{subsec:pre}, we define the attributes of a fresh coin as following. For each fresh coin, $fc_i$, its diversity $dive_i = 1$, $ns_i = 0$ and its diversity $d_i =1$, since there is only one fresh coin in $fc_i$ which is not contained in any RS in $\mathbb{R}_f$. Besides, since $fc_i$ is not contained in any RS in $\mathbb{R}_f$, its $Pr_{max}^i = Pr_{min}^i =0$. \vspace{-1ex}

\begin{lemma}
	\label{lemma:sum_d}
	For the new RS $r_\tau$, its degree $d_\tau$ is the sum of the degree of each super RS and fresh coin which is contained in $r_\tau$, i.e., $d_\tau = \sum_{srs_i \subseteq r_\tau}d_i + \sum_{fc_i \subseteq r_\tau}d_i$.
\end{lemma}\vspace{-3ex}

\begin{proof}
	Suppose $Sub_i$ is the set of RSs in $\mathbb{R}_f$ which are the subsets of $srs_i$ and $Sub_\tau$ is the set of RSs in $\mathbb{R}_f$ which are the subset of $r_\tau$. Thus, $d_\tau = |r_\tau| - |Sub_\tau| = (\sum_{srs_i \subseteq r_\tau}|srs_i| + \sum_{fc_i \subseteq r_\tau}1) - |Sub_\tau|$. By Definition~\ref{def:DS-RRSS}, $\forall srs_i, srs_j \in \mathbb{SRS}$, $Sub_i \cap Sub_j = \emptyset$. Besides, $\forall srs_i \subseteq r_\tau, \forall r \in Sub_i$, $r \in Sub_\tau$, since $r \subseteq srs_i \subseteq r_\tau$. Thus,$|Sub_\tau|$ = $\sum_{srs_i \subseteq r_\tau}ns_i$. Thus, $d_\tau = (\sum_{srs_i \subseteq r_\tau}|srs_i| + \sum_{fc_i \subseteq r_\tau}1) - \sum_{srs_i \subseteq r_\tau}ns_i = (\sum_{srs_i \subseteq r_\tau}|srs_i| -\sum_{srs_i \subseteq r_\tau}ns_i)+ \sum_{fc_i \subseteq r_\tau}1 = \sum_{srs_i \subseteq r_\tau}d_i + \sum_{fc_i \subseteq r_\tau}d_i$.
\end{proof}

As defined in Subsection~\ref{subsec:properites}, for a RS $r_\tau$, $Pr_{max}$ = $\max \{pr_m(c)|$ $c \in r_{\tau}\}$ and $Pr_{min} = \min \{pr_m(c)| c \in r_{\tau} \}$. We will show that when $e^\epsilon\cdot Pr_{min} \cdot (1 - Pr_{max}) - Pr_{max} \cdot (1-Pr_{min}) \geq 0$, the new RS must be a $\epsilon$-CIK-RS (which is defined in Definition~\ref{def:CI-keeping}). Besides, when $e^\epsilon\cdot Pr_{min} \cdot (1 - Pr_{max}) - Pr_{max} \cdot (1-Pr_{min}) \leq 0$, if the degree $d_\tau$ of the new RS $r_\tau$ is smaller than $\frac{(e^\epsilon-1) \cdot (1-Pr_{max}) \cdot (Pr_{min} -1)}{e^\epsilon \cdot \cdot Pr_{min} \cdot (1 - Pr_{max}) - Pr_{max} \cdot (1-Pr_{min})}$, the new RS $r_\tau$ must be a $\epsilon$-CIK-RS.

\begin{lemma}
	\label{lemma:upper}
	For a new RS $r_\tau$, it is a $\epsilon$-CIK-RS, if $e^\epsilon\cdot Pr_{min} \cdot (1 - Pr_{max}) - Pr_{max} \cdot (1-Pr_{min}) \geq 0$. Besides, it is a $\epsilon$-CIK-RS, if $e^\epsilon\cdot Pr_{min} \cdot (1 - Pr_{max}) - Pr_{max} \cdot (1-Pr_{min}) < 0$ and $d_\tau \leq \frac{(e^\epsilon-1) \cdot (1-Pr_{max}) \cdot (Pr_{min} -1)}{e^\epsilon \cdot \cdot Pr_{min} \cdot (1 - Pr_{max}) - Pr_{max} \cdot (1-Pr_{min})}$.
\end{lemma}
\begin{proof}
	As proved in Subsection~\ref{subsec:properites}, to verify if the new RS $r_\tau$ is a $\epsilon$-CIK-RS, we just need to verify if $\frac{(1-pr_{max})\cdot e^\epsilon}{(d_{\tau}-1)\cdot pr_{max}+1} \geq \frac{(1-pr_{min})}{(d_{\tau}-1)\cdot pr_{min}+1}$. Since $(d_\tau - 1) \cdot Pr_{max} + 1$ and $(d_\tau - 1) \cdot Pr_{min} + 1$ must be positive, to verify $\frac{(1-pr_{max})\cdot e^\epsilon}{(d_{\tau}-1)\cdot pr_{max}+1} \geq \frac{(1-pr_{min})}{(d_{\tau}-1)\cdot pr_{min}+1}$ is the same as to verify $e^\epsilon \cdot (1-Pr_{max}) \cdot [(d_\tau - 1) \cdot Pr_{min} + 1] \geq [(d_\tau -1) \cdot Pr_{max} +1] \cdot (1- Pr_{min})$.
	
	Suppose $X = e^\epsilon \cdot Pr_{min} \cdot (1 - Pr_{max}) - Pr_{max} \cdot (1 - Pr_{min})$ and $Y = (e^\epsilon -1) \cdot (1 - Pr_{max}) \cdot (Pr_{min} - 1)$. Thus, to verify if $\frac{(1-pr_{max})\cdot e^\epsilon}{(d_{\tau}-1)\cdot pr_{max}+1} \geq \frac{(1-pr_{min})}{(d_{\tau}-1)\cdot pr_{min}+1}$, we just need to verify if $X \cdot d_\tau \geq Y$. Since $e^\epsilon \geq 1$, $Pr_{max} \leq 1$ and $Pr_{min} \leq 1$, $Y \leq 0$. Thus, if $X \geq 0$, it always hold that $X \cdot d_\tau \geq Y$, which means, $r_\tau$ is a $\epsilon$-CIK-RS. Besides, if $X < 0$ and $ d_\tau \leq \frac{Y}{X}$, it holds that $X \cdot d_\tau \geq \frac{Y}{X}$, which means $r_\tau$ is a $\epsilon$-CIK-RS.
\end{proof}

Besides, suppose $O_{max}$ is the maximal number of coins which are outputted by the the same historical transaction. We show the relationship between the diversity of the new RS $r_\tau$ and the summation of the diversity of each super RS and fresh coin which is contained in $r_\tau$. 

\begin{lemma}
	\label{lemma:diversity}
	For a new RS $r_\tau$, we have $\sum_{srs_i \subseteq r_\tau}dive_i + \sum_{fc_i \subseteq r_\tau}1$ $\geq dive_\tau \geq$ $\sum_{srs_i \subseteq r_\tau}\frac{dive_i}{O_{max}} + \sum_{fc_i \subseteq r_\tau}\frac{1}{O_{max}}$, where $d_\tau$ is the diversity of the new RS $r_\tau$. 
\end{lemma}
\begin{proof}
	In a set of coins, some coins may be outputted by the same historical transaction. Since $O_{max}$ is the maximal number of coins which are outputted by the same historical transaction, it must hold that $dive_\tau \geq \frac{|r_\tau|}{O_{max}}$. Since for any super RS $srs_i$ its diversity $dive_i \leq |srs_i|$, $\frac{dive_i}{O_{max}} \leq \frac{|srs_i|}{O_{max}}$. Therefore, $dive_\tau \geq \frac{|r_\tau|}{O_{max}} = \frac{\sum_{srs_i \subseteq r_\tau}|srs_i|+\sum_{fc_i \subseteq r_\tau}1}{O_{max}} \geq \sum_{srs_i \subseteq r_\tau}\frac{dive_i}{O_{max}} + \sum_{fc_i \subseteq r_\tau}\frac{1}{O_{max}}$.
	
	In addition, denote the set of historical transactions which output coins in $srs_i$ as $HT_i$ and the set of historical transactions which output coins in $r_\tau$ as $HT_\tau$. Since a historical transaction set $HT_i$ may intersect with another historical transaction set $HT_j$ on some historical transactions, $|HT_i| + |HT_j| \geq |HT_i \cup HT_j|$. Thus, $d_\tau = |HT_\tau| = |\bigcup_{srs_i \subseteq r_\tau}HT_i \cup \bigcup_{fc_i}t_i| \leq \sum_{srs_i \subseteq r_\tau}|HT_i| + \sum_{fc_i \subseteq r_\tau}1 = \sum_{srs_i \subseteq r_\tau}dive_i + \sum_{fc_i \subseteq r_\tau}1$.
\end{proof}
	
Thus, by Lemma~\ref{lemma:sum_d} and Lemma~\ref{lemma:upper}, we can transform the $\epsilon$-CIK constraint to the constraint of degree. Specifically, we enumerate all $Pr_{max}$-$Pr_{min}$ pairs and for each pair, we calculate the upper bound of the degree of the new RS, denoted as $\bar{d}$. We require the degree of the new RS cannot exceed $\bar{d}$. Besides, by Lemma~\ref{lemma:diversity}, we can approximately estimate the diversity of the new RS.

Inspired by aforementioned lemmas, we design the Progressive Algorithm. Algorithm~\ref{algo:kp} shows the pseudo-code of our Progressive Algorithm. By the DS constraint, $r_{\tau}$ is composed of some super RS in $\mathbb{SRS}$ or some fresh coins in $\mathbb{F}$. In other words, each super RS and fresh coin is a candidate module of the new RS $r_\tau$. Thus, we get the module set $\mathbb{M}$ by combing $\mathbb{SRS}$ and $\mathbb{F}$ (line 1). By the DS constraint, $r_{\tau}$ has to contain $m_\tau$, where $m_\tau$ is the module in $\mathbb{M}$ that contains $c_\tau$ (line 1). Then we enumerate all $Pr_{max}$-$Pr_{min}$ pairs (line 2 -17), where $m_i$ is the module in the new RS whose $Pr_{max}^i$ is the maximum and $m_j$ is the module in the new RS whose $Pr_{min}^i$ is the minimum. Thus, for the given $Pr_{max}$-$Pr_{min}$ pair, $\omega_{i,j} = m_i \cup m_j \cup m_\tau$ must be contained in the new RS (line 6). Besides, the set of candidate modules, which can be selected in the new RS, is $M_{i,j} = M \backslash (\bigcup_{Pr_{max}^{k} > Pr_{max}^{i}} m_k \cup \omega_{i,j} \cup\bigcup_{Pr_{min}^{k} < Pr_{min}^{j}}m_k)$ (line 7). In addition, we update the diversity of each module in $M_{i,j}$ (line 8). Specifically, the diversity of each module is updated by the number of transactions outputting the coins in modules, excluding the transactions outputting the coins in $\omega_{i,j}$. Then, we calculate the upper bound, $\bar{d}_{i,j}$, of the degree of the new RS and the $O_{max}^{i,j}$, which is the maximal number of coins in $M_{i,j}$ which are outputted by the same transaction (line 9). Then, we run the $\delta$-KP Algorithm (line 10), where the item set is $M_{i,j}$, the weight of the item $m_i$ is $d_i$, the value of the item $m_i$ is $\frac{dive_i}{O_{max}^{i,j}}$, and the capacity of the knapsack is $\bar{d}_{i,j}$. The $\delta$-KP Algorithm is the dynamic programming algorithm~\cite{2013approximation} whose precision parameter is $\delta$. The selection $C_{i,j}$ from the $\delta$-KP Algorithm may violate the budget constraint. Then, we greedily select modules in $C_{i,j}$ to let $\psi_{i,j}$ satisfy the budget constraint (line 11-15). For each module $m_t$ in the selection, we calculate its increase ratio $\rho_t = \frac{dive_{\psi_{i,j}'} - dive_{\psi_{i,j}}}{|m_t|}$, where $\psi_{i,j}' = \psi_{i,j} \cup m_t$. For each iteration, we add the module $m_t$ with the largest $\rho_t$ (14). If the selection $C_{i,j}$ from the $\delta$-KP Algorithm satisfies the budget constraint, we set $\psi_{i,j} = \omega_{i,j} \cup C_{i,j}$ (line 16-17). Finally, we return the $\psi_{i,j}$ with the largest diversity as $r_{\tau}$ (line 18).

\subsection{Theoretic Analyses}

We first  prove the approximate ratio of the $\delta$-KP Algorithm.


\begin{theorem} 
	\label{theorem:kp approximation}
	Suppose $C_{i,j}^*$ is the selection of modules in $M_{i,j}$ whose degree is smaller than $\bar{d}_{i,j}$ and the number of historical transactions outputting coins in $C_{i,j}^*$ is the largest. Denote the diversity of $C_{i,j}^*$ as $dive^*$ and the diversity of $C_{i,j}$ as $dive_C$. It holds that $\frac{dive_C}{dive^*} \geq \frac{1-\delta}{O_{max}}$.
\end{theorem}

\begin{proof}
	Denote the value of each module $m_t$ in $\delta$-KP as $dive_t^\#$. 
	Denote the optimal selection of $\delta$-KP when the input item set is $M_{i,j}$ as $OPT$.
	
	By Lemma~\ref{lemma:diversity}, $dive_C \geq \sum_{m_t \in C_{i,j}}\frac{dive_t}{O_{max}^{i,j}} = \sum_{m_t \in C_{i,j}}dive_t^\#$ and $\sum_{m_t \in C^*_{i,j}}dive_t^\# = \sum_{m_t \in C^*_{i,j}}\frac{dive_i}{O_{max}^{i,j}}\geq \frac{dive^*}{O_{max}^{i,j}}$. Therefore, by~\cite{2013approximation}, $dive_{C} \geq \sum_{m_t \in C_{i,j}}dive_t^\#$ $\geq (1-\delta) \cdot \sum_{m_t \in OPT}dive_t^\#$ $\geq (1-\delta) \cdot \sum_{m_t \in C^*_{i,j}}dive_t^\# \geq (1-\delta) \cdot \frac{dive^*}{O_{max}}$. Therefore, $\frac{dive_{C}}{dive^*} \geq \frac{1-\delta}{O_{max}}$.
\end{proof}

Then, based on Theorem~\ref{theorem:kp approximation}, we give the proof of the approximate ratio of the Progressive Algorithm. 

\begin{theorem}
	\label{theorem:twostep approximation}
	The approximate ratio of the Progressive Algorithm is $\min \{\frac{1-\delta}{O_{max}}, \frac{O_{min}}{O_{max}} \cdot \frac{B-S^+}{B}\}$, where $S^+$ is the maximal size of a module in $\mathbb{M}$.
\end{theorem}


\begin{proof}
	Suppose the RS with maximal diversity is $r_{opt}$ and the RS generated by the Progressive Algorithm is $r_{p}$. Suppose $m_h$ is the module whose $pr_{max}^i$ is the highest in the $r_{opt}$, and $m_s$ is the module whose $pr_{min}^i$ is the smallest in the $r_{opt}$. Denote the number of historical transactions which outputs the coins in $C_{h,s}$ and does not outputs the coins in $\omega_{h,s}$ as $dive_{h,s}$. Besides, denote the diversity of $r_{opt}$ as $dive_{opt}$, the diversity of $r_{p}$ as $dive_{p}$, the diversity of $\omega_{h,s}$ as $dive_{h,s}^\omega$, and the diversity of $\psi_{h,s}$ as $dive_{h,s}^\psi$. Thus, $dive_{p} \geq dive_{h,s}^\psi$.
	
	When $|C_{h,s}| + |\omega_{i,j}| \leq B$, $dive_{h,s}^\psi = dive_{h,s}^\omega + dive_{h,s}^c$. Suppose $OPT'$ is the selection of modules in $M_{h,s}$ whose degree is smaller than $\bar{d}_{h,s}$ and the number of historical  transactions outputting coins in $OPT'$ is the largest. We denote the diversity of $OPT'$ as $dive_{opt}'$. Therefore, $dive_{opt} \leq dive_{opt}'+dive_{h,s}^\omega$. As proved in Theorem~\ref{theorem:kp approximation}, $dive_p \geq dive_{h,s}^c \geq dive_{opt}' \cdot \frac{1-\delta}{O_{max}^{i,j}}$. Therefore, $\frac{dive_{p}}{dive_{opt}} \geq \frac{dive_{h,s}^\omega + dive_{h,s}^c}{dive_{h,s}^\omega + dive_{opt}'} \geq \frac{dive_{h,s}^\omega + dive_{opt}' \cdot \frac{1-\delta}{O_{max}^{h,s}}}{dive_{h,s}^\omega + dive_{opt}'} \geq \frac{1-\delta}{O_{max}^{h,s}} \geq \frac{1-\delta}{O_{max}}.$
	
	When $|C_{h,s}| + |\omega_{i,j}| > B$, $B \geq |\psi_{h,s}| \geq B - S^+$. Since $dive_{h,a}^\psi \geq \frac{|\psi_{h,s}|}{O_{max}}$, $dive_{p} \geq dive_{h,s}^{\psi} \geq \frac{B - S^+}{O_{max}}$. Since $dive_{opt} \leq \frac{|r_{opt}|}{O_{min}}$ and $|r_{opt}| \leq B$, $dive_{opt} \leq \frac{B}{O_{min}}$. Therefore, $\frac{dive_{p}}{dive_{opt}} \geq \frac{O_{min}}{O_{max}} \cdot \frac{B-S^+}{B}$.
\end{proof}

Next, we analyze the time complexity of  Progressive Algorithm.

\begin{theorem}
	\label{theorem:twostep complexity}
	The time complexity of Progressive Algorithm is $\mathcal{O}(\frac{n^5}{\delta})$, where $n = |\mathbb{M}| = |\mathbb{SRS}| + |\mathbb{F}|$ and $\delta$ is the precision parameter of the $\delta$-KP Algorithm.
\end{theorem}

\begin{proof}
	There are $\mathcal{O}(n^2)$ $(i,j)$ pairs and for each pair, the $\delta$-KP Algorithm cost $\mathcal{O}(\frac{n^3}{\delta})$. Thus, the time complexity is $\mathcal{O}(\frac{n^5}{\delta})$. 
\end{proof}
	\section{Game Theoretic Approach}
\label{sec:game}





Although the Progressive Algorithm can generate a new RS with a theoretic guaranteed approximate ratio, its time complexity is high. To solve the CIA-MS-DS problem more efficiently, in this section, we develop an approach based on the game theory. Specifically, we model the CIA-MS-DS problem as a strategic game, where each super RS and each fresh coin corresponds to a player: its goal is to find a strategy that maximizes its own utility. However, to develop such an approach, there are two challenges need to be solved: 1) design utility functions of players to let the sum of their objective is the same as the objective of the CIA-MS-DS problem; and 2) prove the algorithm can achieve a Nash equilibrium with guaranteed quality within polynomial time. We solved these two challenges in the following subsections.

\subsection{Game Theoretic Algorithm}
In this subsection, we solve the first challenge and design utility functions of players.

In strategic games~\cite{armenatzoglou2015real}, players compete with each other to optimize their individual objective functions. Under this framework, each player always tries to choose a strategy that maximizes her/his own utility without taking the effect of her/his choice on the objectives of other players into consideration. The input of the framework is a strategic game, which can be formally represented by a tuple $\langle P, \{S_{p}\}_{p \in P}, \{U_{p}: \times_{{p}\in P}S_{p}\}_{p\in P} 
\rangle$ where $P$ is a set of players and $S_{p}$ represents all the possible strategies that a player $p$ can take during the game to optimize her/his function $U_{p}$. The optimization of $U_{p}$ depends on the own strategy of $p$, as well as the strategies of other players. In \cite{nash1950equilibrium}, Nash points out that a strategic game has a pure Nash equilibrium, if there exists a specific choice of strategies $s_{p} \in S_{p}$ such that the following condition is true for all $p_i \in P$:
$U_{i}(s_1, \cdots, s_{i}, \cdots, s_{|P|}) \leq U_{i}(s_1, \cdots, s'_{i}, \cdots, s_{P}), \forall s'_{i} \in S_{p_i}$.
Thus, no player has the incentive to deviate from her/his current strategy. To express the objective functions of all players, \cite{monderer1996potential} proposes a single function $\Phi: \times_{p\in P}S_{p} 
$, called the potential function in potential games, which constitutes a special class of strategic games. Let $\overline{s_{i}}$ denote the set of strategies followed by all players except $p_i$ (i.e., $\overline{s_{i}} = \{s_1, \cdots, s_{i-1},$ $s_{i+1},\cdots,s_{|P|}\}$). A potential game is exact if there exists a potential function $\Phi$, such that for all $s_i$ and all possible combinations of $\overline{s_i}
$
$U_i(s_i,\overline{s_i}) - U_i(s'_i,\overline{s_i}) = \Phi(s_i,\overline{s_i}) - \Phi(s'_i,\overline{s_i})$. In \cite{monderer1996potential}, it is proved that for potential games, the framework always converges to a pure Nash equilibrium.

We model our problem as a game. 
Each player has two strategies, $s^1$ and $s^0$, which are being selected and not being selected in the new RS, respectively. Given $s_{i}$ and $\overline{s_{i}}$, if $r_{\tau}$ is eligible, the utility of $p_i$ is $U_i(s_i,\overline{s_i})$ $= \frac{|\{t_i| c_i \in r_{\tau}\}|}{|P|}$, otherwise $U_i(s_i,\overline{s_i}) = 0$.
Thus, the objective function of the CIA-MS-DS problem is equal to the summation of the utility of each individual player. The goal of each player is to find the strategy that maximizes its own utility. This decomposition of the objective of the CIA-MS-DS problem into the summation of individual utility functions provides a natural motivation for modeling the CIA-MS-DS problem as a game. Further, we define the potential function as $${\small \Phi(S) =
	\begin{cases}
	\frac{|\{t_i| c_i \in r_{\tau}\}|}{|P|}, & \mbox{$r_{\tau}$ is eligible}\\
	0 & \mbox{otherwise}
	\end{cases}}$$

\begin{algorithm}[t]
	\DontPrintSemicolon
	{
		\KwIn{A spent coin $c_\tau$, a budget $B$, a super RS set $\mathbb{SRS}$, and a fresh coin set $\mathbb{F}$}
		\KwOut{An eligible RS $r_{\tau}$}
		$\mathbb{M} = \mathbb{SRS} \cup \mathbb{F}$, $r_{\tau} = m_\tau$;\;
		\For{$i = 1$ to $|M|$}{
			$P_{i} = M \backslash \bigcup_{m_k \in M \land Pr_{max}^{k} > Pr_{max}^{i}} m_k$;\;
			\For{$j = 1$ to $|M|$}{
				\If{{\scriptsize$m_j \in P_{i}$, $Pr_{max}^{i}\geq Pr_{max}^{\tau}$ and $Pr_{min}^{j} \leq Pr_{min}^{\tau}$}}{
					$\omega_{i,j} = m_i \cup m_j \cup m_\tau$;\;
					{\scriptsize$P_{i,j} = P_i  \backslash (\omega_{i,j} \cup\bigcup_{m_k \in P_i \land Pr_{min}^{k} < Pr_{min}^{j}}m_k)$};\;
					initialize the strategy of each player;\;
					\Repeat{reaching an Nash equilibrium}{
						\ForEach{play $p \in P_{i,j}$}{
							$s_i = s^0$;\;
							\If{the utility of $s^1$ is higher}{
								$s_i = s^1$;\;
							}
						}
					}
					get $\psi_{i,j}$ by strategies of players, combining $\omega_{i,j}$;\;
				}
			}
		}
		\Return{the $\psi_{i,j}$ with maximal diversity as $r_{\tau}$}\;}
		\caption{Game Theoretic Algorithm}
		\label{algo:game}
\end{algorithm}

Algorithm \ref{algo:game} shows the pseudo-code of our Game Theoretic Algorithm. By the DS constraint, $r_{\tau}$ is composed of some super RS in $\mathbb{SRS}$ or some fresh coins in $\mathbb{F}$. In other words, each super RS and fresh coin is a candidate module of the new RS $r_\tau$. Thus, we get the module set $\mathbb{M}$ by combing $\mathbb{SRS}$ and $\mathbb{F}$ (line 1). Besides, by the DS constraint, $r_{\tau}$ has to contain the module $m_\tau$ which contains $c_\tau$ (line 1). Then we enumerate all $Pr_{max}$-$Pr_{min}$ pairs (line 2 -15), where $m_i$ is the module in the new RS whose $Pr_{max}^i$ is the maximum and $m_j$ is the module in the new RS whose $Pr_{min}^i$ is the minimum. Thus, for the given $Pr_{max}$-$Pr_{min}$ pair, $\omega_{i,j} = m_i \cup m_j \cup m_\tau$ must be contained in the new RS (line 6). Besides, for each $Pr_{max}$-$Pr_{min}$ pair, the player set is $P_{i,j} = M \backslash (\bigcup_{Pr_{max}^{k} > Pr_{max}^{i}} m_k \cup \omega_{i,j} \cup\bigcup_{Pr_{min}^{k} < Pr_{min}^{j}}m_k)$ (line 7). For each player, we randomly assign a strategy (line 8). Next, the algorithm starts the best-response procedure (line 9-14). In each iteration, for each player $p$, it selects the strategy with the highest utility. When the two utilities are the same, it selects $s^0$ (line 10-13). When reaching the Nash equilibrium, by strategies of players, we get a candidate RS $\psi_{i,j}$ (line 15). Since $r_{\tau} = m_\tau$ satisfies the required constraints and when $r_{\tau}$ is eligible, the utility of each player is higher, there is at least one eligible RS $\psi_{i,j}$.

\subsection{Theoretic Analyses}

In this subsection, we solve the second challenge and prove the algorithm can achieve a Nash equilibrium with guaranteed quality within polynomial time.

We first prove that for each $Pr_{max}$-$Pr_{min}$ pair, our game is an exact potential game. \vspace{-1ex}

\begin{theorem}
	For each $Pr_{max}$-$Pr_{min}$ pair, the game in the best-response procedure is an exact potential game.
\end{theorem}\vspace{-3ex}

\begin{proof}
	We denote a $Pr_{max}$-$Pr_{min}$ pair by $(i,j)$, where $Pr_{max} = Pr_{max}^i$ and $Pr_{min} = Pr_{min}^j$. For any $(i,j)$ pair, we first proving $U_k(s_k,\overline{s_k}) - U_k(s'_k,\overline{s_k}) = \Phi(s_k,\overline{s_k}) - \Phi(s'_k,\overline{s_k})$.
	
	Suppose $r_{\tau,1}^{i,j}$ is the RS which is generated by the strategies $\overline{s_k}$ and $s_k$, combining $\omega_{i,j}$. Suppose $r_{\tau,2}^{i,j}$ is the RS which is generated by the strategies $\overline{s_k}$ and $s_k'$, combining $\omega_{i,j}$.
	
	When $r_{\tau,1}^{i,j}$ is eligible and $r_{\tau,2}^{i,j}$ is not eligible, $U_k(s_k,\overline{s_k}) = \Phi(s_k,\overline{s_k})$ $= \frac{|\{t_h| c_h \in r_{\tau,1}^{i,j}\}|}{|P_{i,j}|}$ and $U_k(s'_k,\overline{s_k}) = \Phi(s'_k,\overline{s_k}) = 0$. Thus, $U_k(s_k,\overline{s_k})$ $- U_k(s'_k,\overline{s_k}) = \frac{|\{t_h| c_h \in r_{\tau,1}^{i,j}\}|}{|P_{i,j}|} - 0 = \Phi(s_k,\overline{s_k}) - \Phi(s'_k,\overline{s_k})$.
	
	When $r_{\tau,1}^{i,j}$ is eligible and $r_{\tau,2}^{i,j}$ is eligible, $U_k(s_k,\overline{s_k}) = \Phi(s_k,\overline{s_k})$ $= \frac{|\{t_h| c_h \in r_{\tau,1}^{i,j}\}|}{|P_{i,j}|}$ and $U_k(s'_k,\overline{s_k}) = \Phi(s'_k,\overline{s_k}) = \frac{|\{t_h| c_h \in r_{\tau,2}^{i,j}\}|}{|P_{i,j}|}$. Thus, $U_k(s_k,\overline{s_k})$ $- U_k(s'_k,\overline{s_k}) = \frac{|\{t_h| c_h \in r_{\tau,1}^{i,j}\}|}{|P_{i,j}|} - \frac{|\{t_h| c_h \in r_{\tau,2}^{i,j}\}|}{|P_{i,j}|} = \Phi(s_k,\overline{s_k}) - \Phi(s'_k,\overline{s_k})$.
	
	When $r_{\tau,1}^{i,j}$ and $r_{\tau,1}^{i,j}$ are not eligible, $U_k(s_k,\overline{s_k}) = \Phi(s_k,\overline{s_k})$ $= 0$ and $U_k(s'_k,\overline{s_k}) = \Phi(s'_k,\overline{s_k}) = 0$. Thus, $U_k(s_k,\overline{s_k})$ $- U_k(s'_k,\overline{s_k}) = 0 - 0 = \Phi(s_k,\overline{s_k}) - \Phi(s'_k,\overline{s_k})$.
	
	When $r_{\tau,1}^{i,j}$ is not eligible and $r_{\tau,1}^{i,j}$ is eligible, $U_k(s_k,\overline{s_k}) = \Phi(s_k,\overline{s_k})$ $= 0$ and $U_k(s'_k,\overline{s_k}) = \Phi(s'_k,\overline{s_k}) = \frac{|\{t_h| c_h \in r_{\tau,1}^{i,j}\}|}{|P_{i,j}|}$. Thus, $U_k(s_k,\overline{s_k})$ $- U_k(s'_k,\overline{s_k}) = 0 - \frac{|\{t_h| c_h \in r_{\tau,1}^{i,j}\}|}{|P_{i,j}|} = \Phi(s_k,\overline{s_k}) - \Phi(s'_k,\overline{s_k})$.
	
	Then, for any $(i,j)$ pair, by \cite{monderer1996potential}, since $U_k(s_k,\overline{s_k}) - U_k(s'_k,\overline{s_k}) = \Phi(s_k,\overline{s_k}) - \Phi(s'_k,\overline{s_k})$, the game in the best-response procedure is an exact potential game. 
\end{proof}

Since the game in each best-response procedure is an exact potential game, and the set of strategic configurations $S$ is finite, by \cite{monderer1996potential}, a Nash equilibrium can be reached after players changing their strategies a finite number of times. For simplicity, we prove the upper bound for the number of rounds required to reach the convergence of Game Theoretic Algorithm by a scaled version of the problem where the objective function takes integer values. We assume an equivalent game with potential function $\Phi_{\mathbb{Z}} (S) = d \cdot \Phi(S)$ such that $\Phi_{\mathbb{Z}} (S) \in \mathbb{Z}, \forall S$, whcih does not scale with the size of the problem. Then, we can prove the following lemma. 

\begin{lemma}
	The number of rounds required by each best-response procedure to converge to an equilibrium is $\mathcal{O}(d \cdot n)$, where $n = |\mathbb{M}|=|\mathbb{SRS}|+|\mathbb{F}|$.
\end{lemma}\vspace{-3ex}


\begin{proof}
	The scaled version of the Game Algorithm with the potential function $\Phi_{\mathbb{Z}}(S) = d \cdot \Phi(S)$
	will converge to a Nash equilibrium in the same number of rounds as the Game algorithm. Since the change of $\Phi_{\mathbb{Z}} (S)$ is at least 1, and $0 \leq \Phi_{\mathbb{Z}} (S) \leq n$, the number of rounds is at most $\frac{d\cdot n-0}{1} = d \cdot n$.
\end{proof}

Then, we give the proof of the time complexity of the Game Theoretic Algorithm as follows.

\begin{lemma}
	The time complexity of the Game Theoretic Algorithm is $\mathcal{O}(d \cdot n^4)$, where $n = |\mathbb{M}|=|\mathbb{SRS}|+|\mathbb{F}|$.
\end{lemma}\vspace{-3ex}


\begin{proof}
	There are $\mathcal{O}(n^2)$ $Pr_{max}$-$Pr_{min}$ pairs and for each pair, the best-response procedure requires $\mathcal{O}(d \cdot n)$ rounds iteration. Furthermore, in each iteration of the best-response procedure (line 9 - 14), there are $\mathcal{O}(n)$ players. Thus, the time complexity of the Game Theoretic Algorithm is $\mathcal{O}(d \cdot n^4)$.
\end{proof}

After proving Game Theoretic Algorithm can converge within polynomial time, we discuss how good the resulting solution is. Usually, researchers use \textbf{social optimum (OPT)}, \textbf{price of stability(PoS)}, and \textbf{price of anarchy(PoA)} to evaluate the quality of an equilibrium. The $OPT$ is the solution that yields the optimal values to all the objective functions, so that their total utility is maximum. The $PoS$ of a game is the ratio between the best value among its equilibriums and the OPT. The $PoA$ of a game is the ratio between the worst value among its equilibriums and the OPT.

\begin{theorem}
	The $PoS$ is bounded by $\frac{1}{n}$ and the $PoA$ is bounded by $\frac{O_{min} \cdot dive_\tau}{B}$, where $n = |\mathbb{M}|=|\mathbb{SRS}|+|\mathbb{F}|$, $O_{min}$ is the minimal number of coins which are outputted by the same transaction, and $dive_\tau$ is the diversity of the module which contains $c_\tau$.
\end{theorem}\vspace{-3ex}


\begin{proof}
	Let $S_{i,j}$ be the set of strategies of players in $P_{i,j}$ and $U(S_{i,j})$ = $\sum_{k=1}^{|P_{i,j}|}U_k(s_k, \overline{s_k})$. Thus, $\frac{U(S_{i,j})}{|P_{i,j}|} \leq \Phi(S_{i,j}) \leq U(S_{i.j})$. Let $S_{i,j}^*$ be the globally optimal set of strategies of players in $P_{i,j}$ that maximizes $U(S_{i,j})$ and let $OPT_{i,j}$ = $U(S_{i,j}^*)$. Let $S_{i,j}^{\#}$ be the set of strategies of players in $P_{i,j}$ that yields the maximum of $\Phi(S_{i,j})$, i.e., the best Nash equilibrium of a game. 
	Thus, $U(S_{i,j}^{\#}) \geq \Phi(S_{i,j}^{\#})$ $\geq \Phi(S_{i,j}^*) \geq \frac{U(S_{i,j}^*)}{|P_{i,j}|} = \frac{OPT_{i,j}}{|P_{i,j}|}$. Therefore, $PoS \geq \frac{U(S_{i,j}^{\#}) + dive_{i,j}^\omega}{OPT_{i,j}+dive_{i,j}^\omega}$ $\geq \frac{ \frac{1}{|P_{i,j}|} \cdot OPT_{i,j}+ dive_{i,j}^\omega}{OPT_{i,j}+dive_{i,j}^\omega} \geq \frac{1}{|P_{i,j}|} \geq \frac{1}{n}$. Since $m_\tau \subseteq r_{\tau}$ 
	and $OPT \leq \frac{B}{O_{min}}$, $PoA \geq \frac{dive_{\tau}}{OPT} \geq \frac{O_{min} \cdot dive_{\tau}}{B}$.
\end{proof}
	\section{Experimental Study}
\label{sec:experimental}

\begin{table}[t]\vspace{-4ex}
	\begin{center}
		{
			\caption{Experimental Settings (Real).} \label{tab:real} \vspace{1ex}
			\begin{tabular}{l|l}
				{\bf \qquad \qquad \quad Parameters} & {\bf \qquad \qquad \qquad Values} \\ \hline \hline
				the budget $B$ & 40, 60, \textbf{80}, 100, 120\\
				the CI level $\epsilon$ & 1.3, 1.4, \textbf{1.5}, 1.6, 1.7 \\ 
				the range of the degree of & \multirow{2}{*}{[1,9], [1,8], \textbf{[1,7]}, [1,6], [1,5]}\\
				each module $[d^-, d^+]$ & \multirow{2}{*}{~} \\
				the range of $Pr_{max}$ of each & [0.1, 0.5], [0.1, 0.55], \textbf{[0.1, 0.6]},\\
				module, $[PM^-, PM^+]$ & [0.1, 0.65], [0.1, 0.7] \\
				\hline
			\end{tabular}
		}
	\end{center}\vspace{-6ex}
\end{table}

\begin{table}[t]
	\begin{center}
		{
			\caption{Experimental Settings (Synthetic).} \label{tab:synthetic} \vspace{-1.5ex}
			\begin{tabular}{l|l}
				{\bf \qquad \qquad \quad Parameters} & {\bf \qquad \qquad \qquad Values} \\ \hline \hline
				the number, $n$, of modules  & \textbf{50},  60, 70, 80, 90\\
				the number, $o$, of transactions  & 50, 60, \textbf{70}, 80, 90\\
				the budget $B$ & 110, 130, \textbf{150}, 170, 190\\
				the CI level $\epsilon$ & 1.6, 1.7, \textbf{1.8}, 1.9, 2 \\ 
				the range of the degree of & \multirow{2}{*}{\textbf{[1,9]}, [1,8], [1,7], [1,6], [1,5]}\\
				each module $[d^-, d^+]$ & \multirow{2}{*}{~} \\
				the range of the size of  & [11,15], \textbf{[14,18]}, [17,21]\\
				each module $[s^-, s^+]$ & [20,24], [23,27] \\
				the range of the $Pr_{max}$ of & [0.1, 0.2], [0.1, 0.35], \textbf{[0.1, 0.5]},\\
				each module, $[PM^-, PM^+]$ & [0.1, 0.65], [0.1, 0.8] \\
				\hline
			\end{tabular}
		}
	\end{center}\vspace{-3ex}
\end{table}

\begin{figure}[t]\vspace{-2ex}
	\centering
	\includegraphics[scale=0.5]{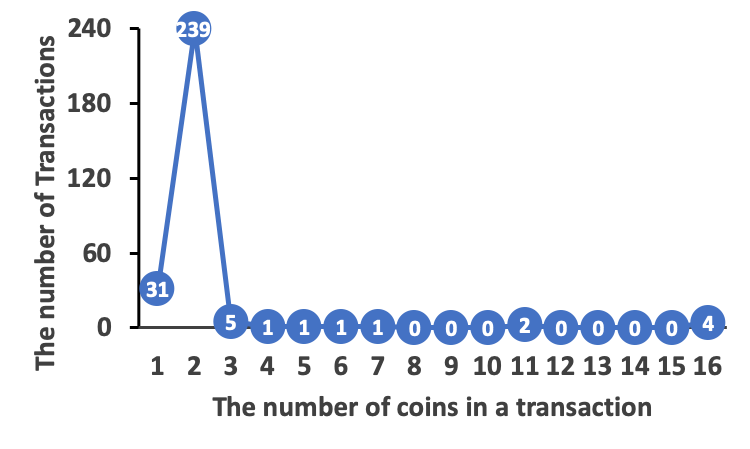}\vspace{-4ex}
	\caption{\small The Distribution of the Number of Coins in a Transaction}\vspace{-1ex}
	\label{fig:source}
\end{figure}

\begin{figure}[t!]\centering
	\subfigure{
		\scalebox{0.3}[0.3]{\includegraphics{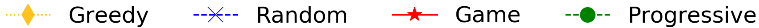}}}\hfill
	\addtocounter{subfigure}{-1}\vspace{-2ex}
	\subfigure[][{\scriptsize Diversity}]{
		\scalebox{0.3}[0.3]{\includegraphics{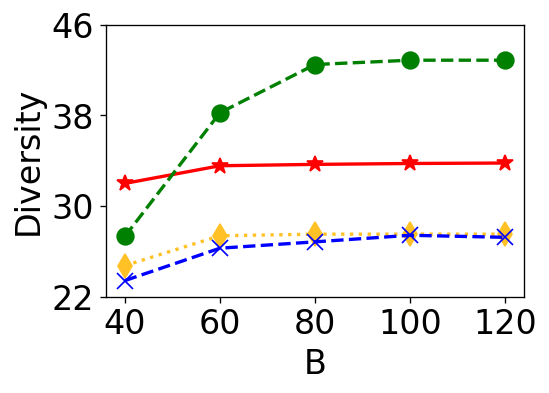}}
		\label{subfig:real_budget_score}}
	\subfigure[][{\scriptsize Running Time}]{
		\scalebox{0.3}[0.3]{\includegraphics{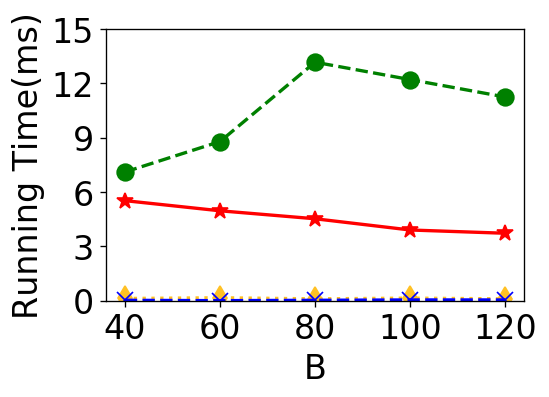}}
		\label{subfig:real_budget_time}}\figureCaptionMargin
	\vspace{-1ex}
	\caption{\small Effect of the Budget (Real)}\figureBelowMargin
	\label{fig:real_budget}
\end{figure}

\vspace{-1ex}
\subsection{Experiment Configuration}

We use both real and synthetic data sets to test our proposed approaches. Specifically, for real data sets, we retrieve the coins in the blocks between 2028242 and 2028273 from the Monero System. The time gap between the block 2028242 and the block 2028273 is one hour. There are 285 transactions and 633 coins. Figure~\ref{fig:source} shows the distribution of the number of coins in a transaction. Most transactions only output two coins and there are 4 transactions which output 16 coins. Since in Monero System, most RS's size is 11, we retrieve 627 coins among them and generate 57 super RSs. For each super RSs, it randomly selects 11 coins as its coin set and its transaction set contains the transactions outputting the selected 11 coins. We uniformly generate the degree of each super SR within the range $[d^-, d^+]$. We generate the $Pr_{max}$ of each super RS within the range $[PM^-, PM^+]$ following the uniform distribution. Since each super RS satisfies the $\epsilon$-CI, we generate the $Pr_{min}$ of each super RS by $\frac{e^\epsilon \cdot (1 - Pr_{max})} { d_i \cdot Pr_{max} +1} = \frac{1- Pr_{min}}{d_i \cdot Pr_{min} + 1}$. We vary the budget from 40 to 120 and the CI level $\epsilon$ from 1.3 to 1.7.

To examine the effects of the number of modules, the number of historical transactions, and each module's size, we generate the synthetic dataset and run the experiments on it. For synthetic data sets, we generate $n$ modules. For each module, we randomly generate its degree, size, mixin set, $Pr_{max}$ and $Pr_{min}$. We uniformly generate the degree of each module within the range $[d^-, d^+]$. We uniformly generate the size of each within the range of $[s^-, s^+]$. We generate the $Pr_{max}$ of each each module within the range $[PM^-, PM^+]$ following the uniform distribution. Since each module satisfies the $\epsilon$-CI, we generate the $Pr_{min}$ of each module by $\frac{e^\epsilon \cdot (1 - Pr_{max})} { d_i \cdot Pr_{max} +1} = \frac{1- Pr_{min}}{d_i \cdot Pr_{min} + 1}$. For each coin, among $o$ historical transactions, we randomly select a historical transaction outputting it. We vary the budget from 110 to 190 and the CI level $\epsilon$ from 1.6 to 2. We vary $n$ from 50 to 90. Since the average size of each module is $16$, 50 modules cover more than 800 coins, which is large enough for real-world applications. In Monero~\cite{Monero}, as shown in the real data sets, the number of coins in an hour is less than 800.

We conduct experiments on both the real data sets and the synthetic data sets to evaluate the effectiveness and efficiency of our two approaches, the Game Theoretic Algorithm and the Progressive Algorithm, in terms of the new RS's diversity and the running time. As proved in Theorem~\ref{theorem:np}, the CIA-MS-DS problem is NP-hard, thus, it is infeasible to calculate the optimal result as the ground truth in large scale datasets. Alternatively, we compare our approaches with two baseline methods, the Greedy Algorithm and the Random Algorithm.  The Greedy Algorithm initializes $r_{\tau}$ as $m_\tau$ and then greedily add the candidate module which can bring the largest increase of the diversity of the temporary $r_{\tau}$ among the modules which would not make the temporary $r_\tau$ ineligible to the temporary $r_\tau$. The Random Algorithm initializes $r_{\tau}$ as $m_\tau$ and then randomly add a candidate module among the modules which would not make the temporary $r_\tau$ ineligible to $r_\tau$.

Table~\ref{tab:real} and Table~\ref{tab:synthetic} introduce our experiment settings on real data sets and synthetic data sets respectively, where the default values of parameters are in bold font. In each set of experiments, we vary one parameter, while setting other parameters to their default values. For each experiment, we sample 50 problem instances and run the algorithms. We report the average value of the running time and the RS's diversity. All our experiments were run on an Intel CPU @2.2 GHz with 16GM RAM in Java.

\subsection{Results on Real Data Sets}

\textbf{Effect of the Budget, $B$.} Figure~\ref{fig:real_budget} illustrates the experimental result on different budgets, $B$, from 40 to 120. In Figure~\ref{subfig:real_budget_score}, when the budget gets larger from 40 to 80, the diversities of the new RSs that generated by four approaches increase; then, they almost keep stable. The reason is that at the beginning, with the increase of $B$, the new RS can contain more mixins. Nevertheless, the new RS is also constrained by the $\epsilon$-CIK constraint. When the budget is large enough, the diversity of the new RS is limited by the $\epsilon$-CIK constraint. And particularly, when the budget is very low, the diversity of the new RS which is generated by the Progressive Algorithm is lower than that of the new RS which is generated by the Game Theoretic Algorithm. Because when the budget is low, the Progressive Algorithm needs to adjust the selection from the $\delta$-KP algorithm, $C_{i,j}$, for the budget constraint (line 11-15). When the budget is very low, the diversity of the new RS which is adjusted by the greedy procedure in the Progressive Algorithm is not as good as that of the Game Algorithm. The reason is that, the greedy procedure is easier to fall into the local optimal trap while the Game Theoretic Algorithm can release this by the games between players.

As shown in Figure~\ref{subfig:real_budget_time}, when the budget increases, the running time of two baseline approaches also increases. The increase of $B$ allows the new RS to contain more mixins, which thus leads to the higher complexity of the CIA-MS-DS problem and the increase of the running time. However, the running time of the Game Theoretic Algorithm decrease. The reason is that, when the budget constraint is relaxed, the game in the Game Theoretic Algorithm is easier to reach the Nash equilibrium. The running time of the Progressive Algorithm increases at the beginning. Because when the budget increases but still is a strict constraint, the Progressive Algorithm spends more time on the greedy procedure. But later, when the budget is large enough, the running time of the Progressive Algorithm decreases. The reason is that, when the budget is large enough, the budget constraint is relaxed and there are more $C_{i,j}$ whose size are smaller than the budget and the Progressive Algorithm do not need to run the greedy procedure for these candidate RSs.

\begin{figure}[t!]\centering \vspace{-3ex}
	\subfigure{
		\scalebox{0.3}[0.3]{\includegraphics{figures/bar.png}}}\hfill
	\addtocounter{subfigure}{-1}\vspace{-2ex}
	\subfigure[][{\scriptsize Diversity}]{
		\scalebox{0.3}[0.3]{\includegraphics{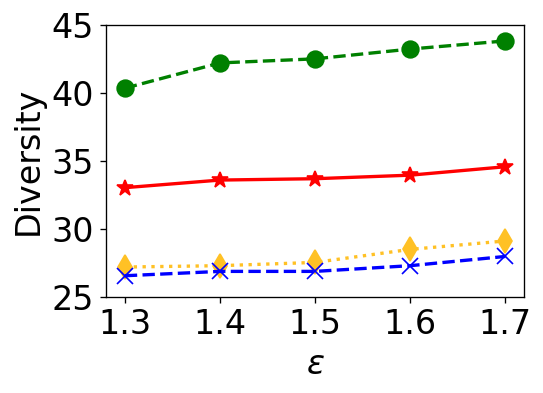}}
		\label{subfig:real_level_score}}
	\subfigure[][{\scriptsize Running Time}]{
		\scalebox{0.3}[0.3]{\includegraphics{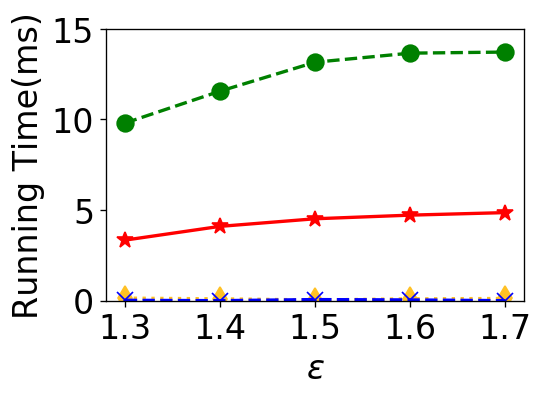}}
		\label{subfig:real_level_time}}\figureCaptionMargin
	\vspace{-1ex}
	\caption{\small Effect of the CI Level (Real)}\figureBelowMargin
	\label{fig:real_level}
\end{figure}

\begin{figure}[t!]\centering
	\subfigure{
		\scalebox{0.3}[0.3]{\includegraphics{figures/bar.png}}}\hfill
	\addtocounter{subfigure}{-1}\vspace{-2ex}
	\subfigure[][{\scriptsize Diversity}]{
		\scalebox{0.3}[0.3]{\includegraphics{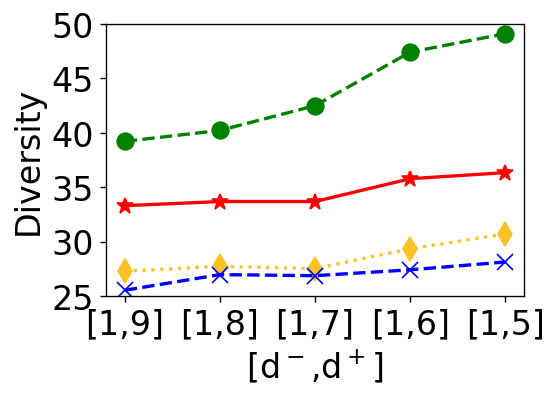}}
		\label{subfig:real_degree_score}}
	\subfigure[][{\scriptsize Running Time}]{
		\scalebox{0.3}[0.3]{\includegraphics{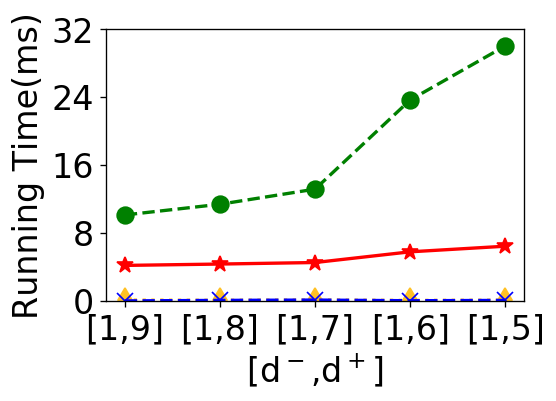}}
		\label{subfig:real_degree_time}}\figureCaptionMargin
	\vspace{-1ex}
	\caption{\small Effect of the Range of the Degree of each Super RS (Real)}\figureBelowMargin
	\label{fig:real_degree}
\end{figure}

\textbf{Effect of the CI Level, $\epsilon$.} Figure~\ref{fig:real_level} illustrates the experimental result on different CI levels, $\epsilon$, from 1.3 to 1.7. In Figure~\ref{subfig:real_level_score}, when the $\epsilon$ increases, the diversities of the new RSs that generated by four approaches also increase. The reason is that the increase of $\epsilon$ makes the $\epsilon$-CIK constraint more relaxed and the new RS has more valid modules. As shown in Figure~\ref{subfig:real_level_time}, when the $\epsilon$ increases, the running time of four approaches also increases. Because, the increase of $\epsilon$ let the new RS has more valid modules, which thus leads to the higher complexity of the CIA-MS-DS problem and the increase of the running time.

\textbf{Effect of the Range of the Degree of each Super RS, $[d^-, d^+]$.} Figure~\ref{fig:real_degree} illustrates the experimental results on different ranges, $[d^-, d^+]$, of the degree of each super RS, from [1,9] to [1,5]. In Figure~\ref{subfig:real_degree_score}, the diversities of the new RSs that generated by our four approaches increase, when the average value of degrees of modules decreases. The reason is that, when the average value of degrees of modules is higher, it is more difficult to satisfy the $\epsilon$-CIK constraint. In other words, the new RS has fewer valid modules. In Figure~\ref{subfig:real_degree_time}, the running time of our four approaches increases when the average value of degrees of modules decreases. Because when the average value of degrees of modules increases, the new RS has fewer valid modules, which thus leads to lower complexity of the CIA-MS-DS problem and the decrease of the running time.

\begin{figure}[t!]\centering\vspace{-3ex}
	\subfigure{
		\scalebox{0.3}[0.3]{\includegraphics{figures/bar.png}}}\hfill
	\addtocounter{subfigure}{-1}\vspace{-2ex}
	\subfigure[][{\scriptsize Diversity}]{
		\scalebox{0.3}[0.3]{\includegraphics{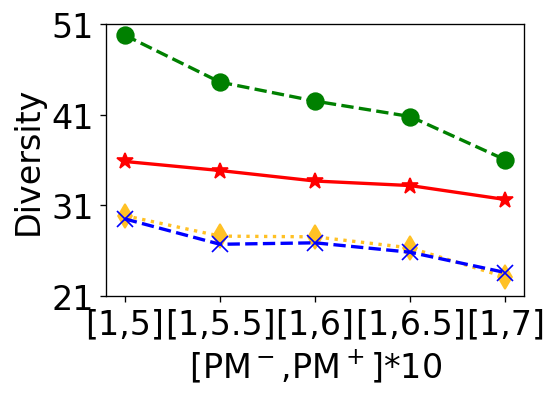}}
		\label{subfig:real_pr_score}}
	\subfigure[][{\scriptsize Running Time}]{
		\scalebox{0.3}[0.3]{\includegraphics{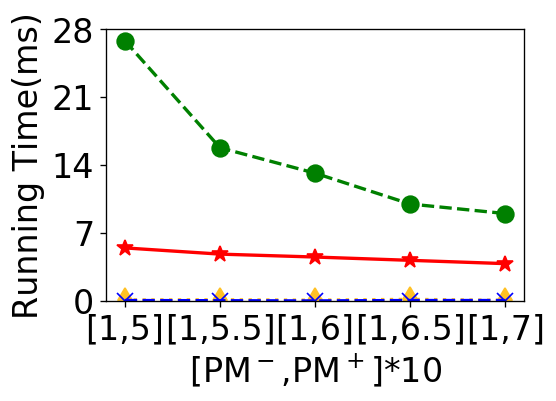}}
		\label{subfig:real_pr_time}}\figureCaptionMargin
	\vspace{-1ex}
	\caption{\small Effect of the Range of the $Pr_{max}$ of each Super RS (Real)}\figureBelowMargin
	\label{fig:real_pr}
\end{figure}

\begin{figure}[t!]\centering
	\subfigure{
		\scalebox{0.3}[0.3]{\includegraphics{figures/bar.png}}}\hfill
	\addtocounter{subfigure}{-1}\vspace{-2ex}
	\subfigure[][{\scriptsize Diversity}]{
		\scalebox{0.3}[0.311]{\includegraphics{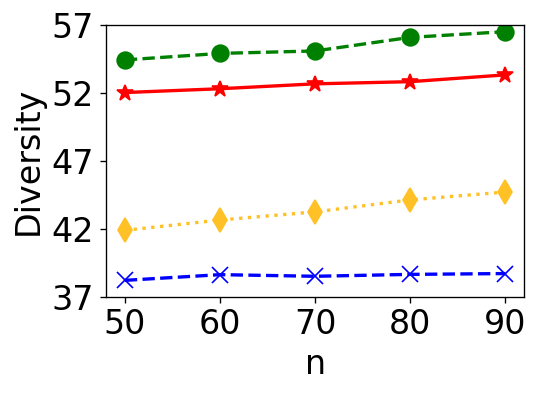}}
		\label{subfig:module_score}}
	\subfigure[][{\scriptsize Running Time}]{
		\scalebox{0.3}[0.311]{\includegraphics{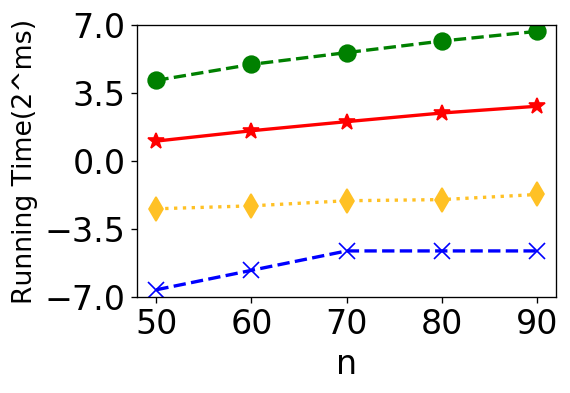}}
		\label{subfig:module_time}}\figureCaptionMargin
	\vspace{-1ex}
	\caption{\small Effect of the Number of Modules (Synthetic)}\figureBelowMargin
	\label{fig:module}
\end{figure}

\textbf{Effect of the Range of the $Pr_{max}$ of each Super RS, $[PM^-$, $PM^+]$.} Figure~\ref{fig:real_pr} illustrates the experimental results on different ranges, $[PM^-, PM^+]$, of the $Pr_{max}$ of each module, from [0.1,0.5] to [0.1,0.7]. In Figure~\ref{subfig:real_pr_score}, the diversities of the new RSs that generated by our four approaches decrease, when the range is wider. The reason is that, when the range of the $Pr_{max}$ of each module is wider, it is more difficult to satisfy the $\epsilon$-CIK constraint. In other words, the new RS has fewer valid modules. In Figure~\ref{subfig:real_pr_time}, the running time of our four approaches decreases when the difference of the $Pr_{max}$ of modules increases. The reason is that, when the range of the $Pr_{max}$ of modules is narrower, the new RS has more valid modules, which thus leads to higher complexity of the CIA-MS-DS problem and the increase of the running time. Specifically, compared with the Game Theoretic Algorithm, the Progressive Algorithm is more sensitive about the change of the range. The reason is that, when the range is wider, the cardinality of $M_{i,j}$ is smaller and the running time of the $\delta$-KP Algorithm is smaller. 

\subsection{Results on Synthetic Data Sets}

To examine the effects of the number of modules, the number of historical transactions, and each module's size, we generate the synthetic dataset and run the experiments on it. We also test the effects of the budget, the CI level, the range of the degree of each module, and the range of the $Pr_{max}$ of each module on the synthetic data sets. Due to the space limitation, please refer to Appendix F of our technical report~\cite{report} for details.

\textbf{Effect of the Number, $n$ of Modules.} Figure~\ref{fig:module} illustrates the experimental result on a different number, $n$, of modules from 50 to 90. In Figure~\ref{subfig:module_score}, when the number of modules increases, the diversities of the new RSs that generated by four approaches also increase. The reason is that the increase of $n$ let the new RS has more valid modules. Specifically, our two approximate algorithms achieve better results than the two baseline algorithms. The RS which is generated by the Progressive Algorithm has the largest diversity. As shown in Figure~\ref{subfig:module_time}, when the number of modules increases, the running time of four approaches increases. When the number of modules increases, the new RS has more valid modules, which thus lead to the increase of the running time. Specifically, the Progressive Algorithm's running time is the highest while the running time of the Game Theoretic Algorithm is much lower.

\textbf{Effect of the Number, $o$ of Historical Transactions.} Figure~\ref{fig:transaction} illustrates the experimental result on a different number, $o$, of historical transactions from 60 to 100. In Figure~\ref{subfig:transaction_score}, when the number of historical transactions increases, the diversities of new RSs that generated by four approaches also increase. The reason is that the increase of $o$ decreases the overlap between historical transaction set of modules, where the historical transaction set of a module is the set of historical transaction outputting the coins in the module. As shown in Figure~\ref{subfig:transaction_time}, when the number of historical transactions increases, the running time of four approaches almost keeps stable. Because the complexity of the CIA-MS-DS problem and the time complexities of four approaches are all not related to the number of historical transactions.

\begin{figure}[t!]\centering\vspace{-3ex}
	\subfigure{
		\scalebox{0.3}[0.3]{\includegraphics{figures/bar.png}}}\hfill
	\addtocounter{subfigure}{-1}\vspace{-2ex}
	\subfigure[][{\scriptsize Diversity}]{
		\scalebox{0.3}[0.3]{\includegraphics{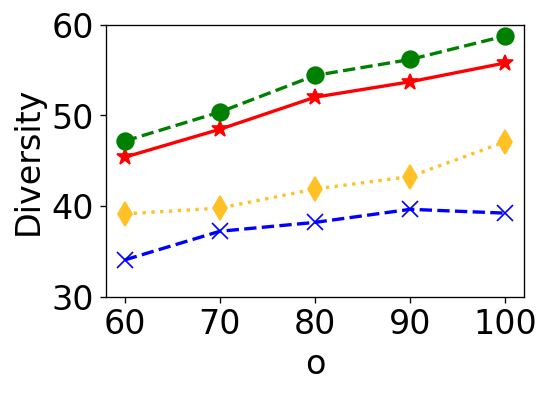}}
		\label{subfig:transaction_score}}
	\subfigure[][{\scriptsize Running Time}]{
		\scalebox{0.3}[0.3]{\includegraphics{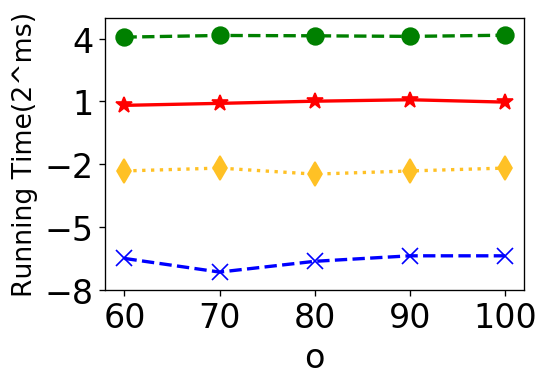}}
		\label{subfig:transaction_time}}\figureCaptionMargin
	\vspace{-1ex}
	\caption{\scriptsize Effect of the Number of Historical Transactions (Synthetic)}\figureBelowMargin
	\label{fig:transaction}
\end{figure}

\textbf{Effect of the Range of the Size of each Module, $[s^-, s^+]$.} Figure~\ref{fig:size} illustrates the experimental results on different ranges, $[s^-, s^+]$, of the size of each module, from [11,15] to [23,27]. In Figure~\ref{subfig:size_score}, the diversities of the new RSs that generated by our four approaches decrease, when the size of each module increases. The reason is that, when the average value of sizes of modules is higher, it is more difficult to satisfy the budget constraint. In other words, the new RS has fewer valid modules. In Figure~\ref{subfig:size_time}, the running time of our Progressive Algorithm and two baseline approaches decreases when the average value of sizes of modules increases. The reason is that, when the average value of sizes of modules is higher, the new RS has fewer valid modules, which thus leads to lower complexity of the CIA-MS-DS problem and the decrease of the running time. However, the running time of our Game Theoretic Algorithm increases when the average value of sizes of modules increases. The reason is that, when the average value of sizes of modules is higher, the Game Theoretic Algorithm needs to do more games to reach the Nash equilibrium.

We finally summarize our findings as following: \vspace{-1ex}

\begin{itemize}[leftmargin=*]
	\item Our two approximate algorithms can achieve results with higher diversity compared with that of two baselines. \vspace{-1ex}
	\item The Progressive Algorithm outputs the RS with the highest diversity while it costs much time. But its speed is still acceptable for the current public blockchain systems, like Monero~\cite{Monero}. \vspace{-1ex}
	\item The Game Theoretic Algorithm outputs the RS with high diversity quickly. Compared with the Progressive Algorithm, it is suitable for consortium blockchain systems with high TPSs. \vspace{-1ex}
\end{itemize}

\vspace{-3ex}
	\section{Related Work}
\label{sec:related}

Blockchain technologies are gaining massive momentum in recent years, largely due to its immutability and transparency. Many applications for security trading and settlement~\cite{Ripple}, asset and finance management~\cite{Melonport}~\cite{morgan2016unlocking}, banking and insurance~\cite{gs2016blockchain} are evaluated. However, the transparency character also brings the privacy problem. In many practice applications, users do not want to share all their information with other participants. To solve the privacy problem, some researchers have proposed some privacy-preserved blockchain systems. 

These works can be classified into two categories. The works in the first category focus on developing mixing protocols~\cite{maxwell2013coinjoin, ruffing2014coinshuffle,ruffing2017p2p,moreno2017pathshuffle}. In these protocols, anonymous service providers use mixing protocols to confuse the trails of transactions. The client's funds are divided into smaller parts which are mixed randomly with similar random parts of other clients. This helps to break links between the users and the transactions they purchased while these methods rely on mixers' trustiness and the mixers know the transactions' privacy. These methods weaken the blockchain system's decentration.

The works in the second category focus on developing advanced encryption methods. 
In the Monero blockchain system~\cite{yuen2019ringct}, the researchers build the privacy-preserved blockchain system based on the Ring Confidential Transaction (RingCT)~\cite{van2013cryptonote}~\cite{sun2017ringct}~\cite{yuen2019ringct}. In the first version of RingCT~\cite{van2013cryptonote} (RingCT 1.0), the researchers adapted a RS scheme~\cite{cramer1994proofs} to protect the transaction's sender's identity. In the second version of RingCT~\cite{sun2017ringct}, the researchers put forward a new efficient RingCT protocol (RingCT 2.0), which is built upon the well-known Pedersen commitment~\cite{pedersen1991non} and saves almost half RS's size compared with former version RingCT. In ~\cite{yuen2019ringct}, the researchers put forward the newest RingCT protocol based on Bulletproof~\cite{bunz2018bulletproofs}, which decreases the RS's size from $\mathcal{O}(n)$ to $\mathcal{O}(\log n)$. 

While these cryptographic techniques are used to some extent achieve confidentiality, the considerable overhead of such techniques makes them impractical~\cite{androulaki2018hyperledger}. Besides, these techniques assume the adversary has no extra information. However, since all data in the blockchain system is accessed, adversaries can attack a user's privacy by analyzing the traffic flow on the blockchain system. Our work considers the traffic flow's impact and proposes methods to strengthen RSs'  privacy-preserving effect by selecting a set of desirable mixins.

	\section{Conclusion}
\label{sec:conclusion}

\begin{figure}[t!]\centering\vspace{-3ex}
	\subfigure{
		\scalebox{0.3}[0.3]{\includegraphics{figures/bar.png}}}\hfill
	\addtocounter{subfigure}{-1}\vspace{-2ex}
	\subfigure[][{\scriptsize Diversity}]{
		\scalebox{0.3}[0.3]{\includegraphics{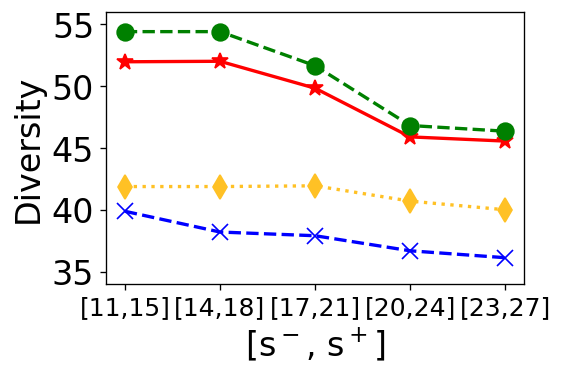}}
		\label{subfig:size_score}}
	\subfigure[][{\scriptsize Running Time}]{
		\scalebox{0.3}[0.3]{\includegraphics{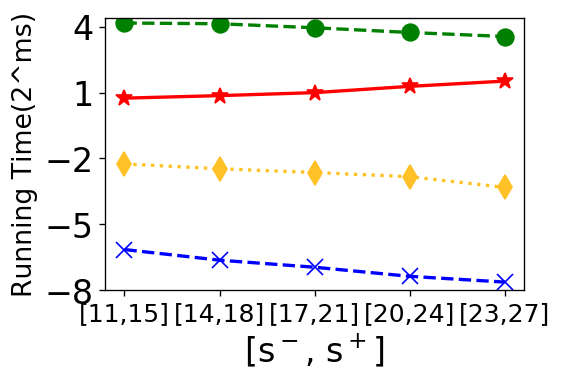}}
		\label{subfig:size_time}}\figureCaptionMargin
	\vspace{-1ex}
	\caption{\small Effect of the Range of the Size of each Module (Synthetic)}\figureBelowMargin
	\label{fig:size}
\end{figure}

In this paper, we formulate a differential privacy definition, namely $\epsilon$-coin-indistinguishability, in blockchain scenarios. We show that, if a RS satisfies the $\epsilon$-coin-indistinguishability, it is resistant to the ``Chain-Reaction" analysis. Besides, in this paper, we formulate the coin-indistinguishability-aware mixin selection problem with the disjoint-superset constraint (CIA-MS-DS), which aims to find a set of mixins which satisfies the $\epsilon$-coin-indistinguishability constraint, as well as the budget constraint, and has the maximal diversity. We formally prove that the CIA-MS-DS problem has some significant properties which can help to simplify the problem, while the CIA-MS-DS problem still is an NP-hard problem. To efficiently and effectively solve the CIA-MS-DS problem, we propose a novel framework, CoinMagic, and propose two approximate algorithms, namely the Progressive Algorithm and the Game Algorithm, with theoretical guarantees. When evaluated on the real and synthetic data sets, our approaches achieved clearly better performance than two baseline algorithms.

	\bgroup\small
	\bibliographystyle{ieeetr}
	\let\xxx=\bibitem\def\bibitem{\par\vspace{1mm}\xxx} 
	\bibliography{references/add}

\begin{thebibliography}{10}

\bibitem{nakamoto2008bitcoin}
S.~Nakamoto {\em et~al.}, ``Bitcoin: A peer-to-peer electronic cash system,''
  2008.

\bibitem{wood2014ethereum}
G.~Wood {\em et~al.}, ``Ethereum: A secure decentralised generalised
  transaction ledger,'' {\em Ethereum project yellow paper}, vol.~151,
  no.~2014, pp.~1--32, 2014.

\bibitem{linkchain}
``[online] \text{Thunderchain}.'' \url{https://www.lianxiangcloud.com}.

\bibitem{xu2019vchain}
C.~Xu, C.~Zhang, and J.~Xu, ``vchain: Enabling verifiable boolean range queries
  over blockchain databases,'' in {\em Proceedings of the 2019 International
  Conference on Management of Data}, pp.~141--158, ACM, 2019.

\bibitem{zhang2019gem}
C.~Zhang, C.~Xu, J.~Xu, Y.~Tang, and B.~Choi, ``Gem\^{} 2-tree: A gas-efficient
  structure for authenticated range queries in blockchain,'' in {\em 2019 IEEE
  35th International Conference on Data Engineering (ICDE)}, pp.~842--853,
  IEEE, 2019.

\bibitem{korpela2017digital}
K.~Korpela, J.~Hallikas, and T.~Dahlberg, ``Digital supply chain transformation
  toward blockchain integration,'' in {\em proceedings of the 50th Hawaii
  international conference on system sciences}, 2017.

\bibitem{azaria2016medrec}
A.~Azaria, A.~Ekblaw, T.~Vieira, and A.~Lippman, ``Medrec: Using blockchain for
  medical data access and permission management,'' in {\em 2016 2nd
  International Conference on Open and Big Data (OBD)}, pp.~25--30, IEEE, 2016.

\bibitem{guo2016blockchain}
Y.~Guo and C.~Liang, ``Blockchain application and outlook in the banking
  industry,'' {\em Financial Innovation}, vol.~2, no.~1, p.~24, 2016.

\bibitem{MVL}
``[online] \text{Mass vehicle ledger}.'' \url{https://mvlchain.io/}.

\bibitem{andres2012geo}
M.~E. Andr{\'e}s, N.~E. Bordenabe, K.~Chatzikokolakis, and C.~Palamidessi,
  ``Geo-indistinguishability: Differential privacy for location-based
  systems,'' {\em arXiv preprint arXiv:1212.1984}, 2012.

\bibitem{okamoto1991universal}
T.~Okamoto and K.~Ohta, ``Universal electronic cash,'' in {\em Annual
  international cryptology conference}, pp.~324--337, Springer, 1991.

\bibitem{van2013cryptonote}
N.~Van~Saberhagen, ``Cryptonote v 2.0,'' 2013.

\bibitem{sun2017ringct}
S.-F. Sun, M.~H. Au, J.~K. Liu, and T.~H. Yuen, ``Ringct 2.0: A compact
  accumulator-based (linkable ring signature) protocol for blockchain
  cryptocurrency monero,'' in {\em European Symposium on Research in Computer
  Security}, pp.~456--474, Springer, 2017.

\bibitem{yuen2019ringct}
T.~H. Yuen, S.-f. Sun, J.~K. Liu, M.~H. Au, M.~F. Esgin, Q.~Zhang, and D.~Gu,
  ``Ringct 3.0 for blockchain confidential transaction: Shorter size and
  stronger security,'' 2019.

\bibitem{Monero}
``[online] \text{Monero}.'' \url{https://www.getmonero.org/}.

\bibitem{Bytecoin}
``[online] \text{Bytecoin}.'' \url{https://bytecoin.org/}.

\bibitem{sweeney2002k}
L.~Sweeney, ``k-anonymity: A model for protecting privacy,'' {\em International
  Journal of Uncertainty, Fuzziness and Knowledge-Based Systems}, vol.~10,
  no.~05, pp.~557--570, 2002.

\bibitem{li2007t}
N.~Li, T.~Li, and S.~Venkatasubramanian, ``t-closeness: Privacy beyond
  k-anonymity and l-diversity,'' in {\em 2007 IEEE 23rd International
  Conference on Data Engineering}, pp.~106--115, IEEE, 2007.

\bibitem{moser2018empirical}
M.~M{\"o}ser, K.~Soska, E.~Heilman, K.~Lee, H.~Heffan, S.~Srivastava, K.~Hogan,
  J.~Hennessey, A.~Miller, A.~Narayanan, {\em et~al.}, ``An empirical analysis
  of traceability in the monero blockchain,'' {\em Proceedings on Privacy
  Enhancing Technologies}, vol.~2018, no.~3, pp.~143--163, 2018.

\bibitem{chervinski2019floodxmr}
J.~O.~M. Chervinski, D.~Kreutz, and J.~Yu, ``Floodxmr: Low-cost transaction
  flooding attack with monero's bulletproof protocol.,'' {\em IACR Cryptology
  ePrint Archive}, vol.~2019, p.~455, 2019.

\bibitem{2013approximation}
V.~V. Vazirani, {\em Approximation algorithms}.
\newblock Springer Science \& Business Media, 2013.

\bibitem{dwork2011differential}
C.~Dwork, ``Differential privacy,'' {\em Encyclopedia of Cryptography and
  Security}, pp.~338--340, 2011.

\bibitem{hinteregger2019short}
A.~Hinteregger and B.~Haslhofer, ``Short paper: An empirical analysis of monero
  cross-chain traceability,'' in {\em International Conference on Financial
  Cryptography and Data Security}, pp.~150--157, Springer, 2019.

\bibitem{armenatzoglou2015real}
N.~Armenatzoglou, H.~Pham, V.~Ntranos, D.~Papadias, and C.~Shahabi, ``Real-time
  multi-criteria social graph partitioning: A game theoretic approach,'' in
  {\em ACM SIGMOD}, pp.~1617--1628, 2015.

\bibitem{nash1950equilibrium}
J.~F. Nash {\em et~al.}, ``Equilibrium points in n-person games,'' {\em PNAS},
  vol.~36, no.~1, pp.~48--49, 1950.

\bibitem{monderer1996potential}
D.~Monderer and L.~S. Shapley, ``Potential games,'' {\em Games and economic
  behavior}, vol.~14, no.~1, pp.~124--143, 1996.

\bibitem{report}
``[online] \text{Technical Report}.''
  \url{https://cspcheng.github.io/pdf/CoinMagic.pdf}.

\bibitem{Ripple}
``[online] \text{Ripple}.'' \url{https://ripple.com}.

\bibitem{Melonport}
``[online] \text{Blockchain software for asset management}.''
  \url{http://melonport.com}.

\bibitem{morgan2016unlocking}
J.~Morgan and O.~Wyman, ``Unlocking economic advantage with blockchain. a guide
  for asset managers.,'' {\em New York: JP Morgan Reports}, 2016.

\bibitem{gs2016blockchain}
G.~Group {\em et~al.}, ``Blockchain: Putting theory into practice,'' 2016.

\bibitem{maxwell2013coinjoin}
G.~Maxwell, ``Coinjoin: Bitcoin privacy for the real world,'' in {\em Post on
  Bitcoin forum}, 2013.

\bibitem{ruffing2014coinshuffle}
T.~Ruffing, P.~Moreno-Sanchez, and A.~Kate, ``Coinshuffle: Practical
  decentralized coin mixing for bitcoin,'' in {\em European Symposium on
  Research in Computer Security}, pp.~345--364, Springer, 2014.

\bibitem{ruffing2017p2p}
T.~Ruffing, P.~Moreno-Sanchez, and A.~Kate, ``P2p mixing and unlinkable bitcoin
  transactions.,'' in {\em NDSS}, 2017.

\bibitem{moreno2017pathshuffle}
P.~Moreno-Sanchez, T.~Ruffing, and A.~Kate, ``Pathshuffle: Credit mixing and
  anonymous payments for ripple,'' {\em Proceedings on Privacy Enhancing
  Technologies}, vol.~2017, no.~3, pp.~110--129, 2017.

\bibitem{cramer1994proofs}
R.~Cramer, I.~Damg{\aa}rd, and B.~Schoenmakers, ``Proofs of partial knowledge
  and simplified design of witness hiding protocols,'' in {\em Annual
  International Cryptology Conference}, pp.~174--187, Springer, 1994.

\bibitem{pedersen1991non}
T.~P. Pedersen, ``Non-interactive and information-theoretic secure verifiable
  secret sharing,'' in {\em Annual international cryptology conference},
  pp.~129--140, Springer, 1991.

\bibitem{bunz2018bulletproofs}
B.~B{\"u}nz, J.~Bootle, D.~Boneh, A.~Poelstra, P.~Wuille, and G.~Maxwell,
  ``Bulletproofs: Short proofs for confidential transactions and more,'' in
  {\em 2018 IEEE Symposium on Security and Privacy (SP)}, pp.~315--334, IEEE,
  2018.

\bibitem{androulaki2018hyperledger}
E.~Androulaki, A.~Barger, V.~Bortnikov, C.~Cachin, K.~Christidis, A.~De~Caro,
  D.~Enyeart, C.~Ferris, G.~Laventman, Y.~Manevich, {\em et~al.}, ``Hyperledger
  fabric: a distributed operating system for permissioned blockchains,'' in
  {\em Proceedings of the Thirteenth EuroSys Conference}, p.~30, ACM, 2018.

\end{thebibliography}
	\egroup
	
	\appendix

\section{Results on Synthetic Data Sets}

\noindent\textbf{Effect of the Budget, $B$.} Figure~\ref{fig:budget} illustrates the experimental result on different budgets, $B$, from 110 to 190. In Figure~\ref{subfig:budget_score}, when the budget gets larger from 110 to 170, the diversities of the new ring signatures that generated by the four approaches increase; then, they almost keep stable. The reason is that at the beginning, with the increase of $B$, the new ring signature can contain more mixins. Nevertheless, the new ring signature is also constrained by the $\epsilon$-CIK constraint. When the budget is large enough, the new ring signature's diversity is limited by the $\epsilon$-CIK constraint. As shown in Figure~\ref{subfig:budget_time}, when the budget increases, the running time of four approaches also increases. Because the increase of $B$ allows the new ring signature to contain more mixins, which thus leads to the higher complexity of the CIA-MS-DS problem and the increase of the running time. However, when the budget gets larger from 170 to 190, the running time of our Game Algorithm and Progressive Algorithm both decrease. The reason is that, when the budget constraint is relaxed, the game in the Game Algorithm is easier to reach the Nash equilibrium and the Progressive Algorithm costs less time to alter the ring signature after the $\delta$-KP for the budget constraint.

\begin{figure}[t!]\centering
	\subfigure{
		\scalebox{0.3}[0.3]{\includegraphics{figures/bar.png}}}\hfill
	\addtocounter{subfigure}{-1}
	\subfigure[][{\scriptsize Diversity}]{
		\scalebox{0.3}[0.3]{\includegraphics{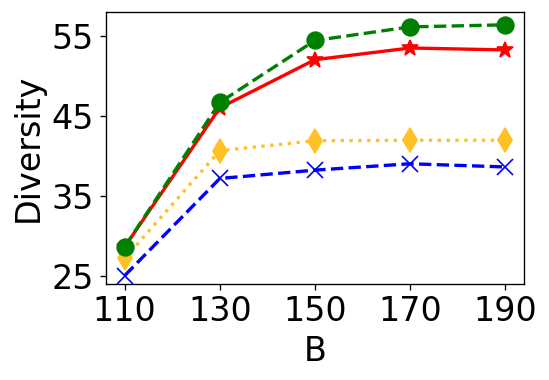}}
		\label{subfig:budget_score}}
	\subfigure[][{\scriptsize Running Time}]{
		\scalebox{0.3}[0.3]{\includegraphics{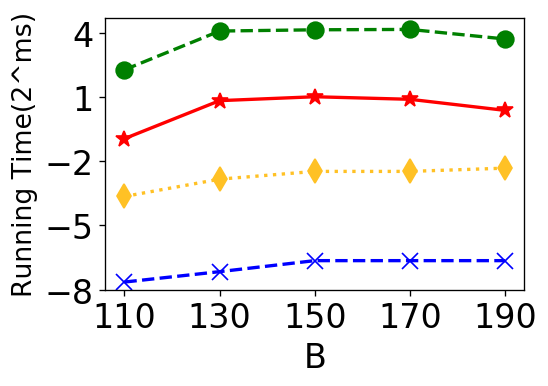}}
		\label{subfig:budget_time}}\figureCaptionMargin
	\caption{\small Effect of the Budget (Synthetic)}\figureBelowMargin
	\label{fig:budget}
\end{figure}

\begin{figure}[t!]\centering
	\subfigure{
		\scalebox{0.3}[0.3]{\includegraphics{figures/bar.png}}}\hfill
	\addtocounter{subfigure}{-1}
	\subfigure[][{\scriptsize Diversity}]{
		\scalebox{0.3}[0.3]{\includegraphics{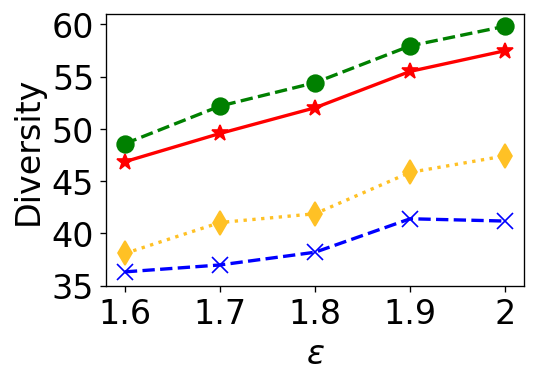}}
		\label{subfig:level_score}}
	\subfigure[][{\scriptsize Running Time}]{
		\scalebox{0.3}[0.3]{\includegraphics{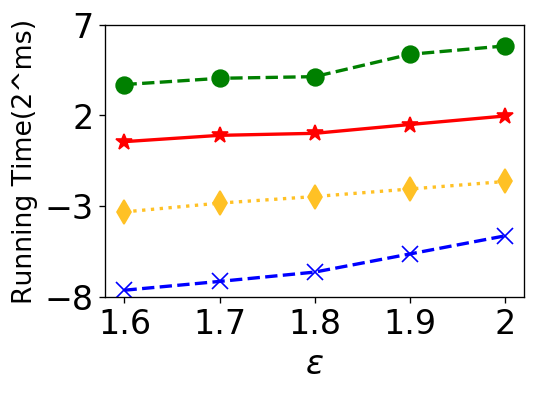}}
		\label{subfig:level_time}}\figureCaptionMargin
	\caption{\small Effect of the CI Level (Synthetic)}\figureBelowMargin
	\label{fig:level}
\end{figure}

\noindent\textbf{Effect of the CI Level, $\epsilon$.} Figure~\ref{fig:level} illustrates the experimental result on different CI levels, $\epsilon$, from 1.6 to 2. In Figure~\ref{subfig:level_score}, when the $\epsilon$ increases, the diversities of the new ring signatures that generated by the four approaches also increase. The reason is that the increase of $\epsilon$ makes the $\epsilon$-CIK constraint more relaxed and the new ring signature has more valid modules. As shown in Figure~\ref{subfig:level_time}, when the $\epsilon$ increases, the running time of four approaches also increases. Because the increase of $\epsilon$ let the new ring signature has more valid modules, which thus leads to the higher complexity of the CIA-MS-DS problem and the increase of the running time.

\noindent\textbf{Effect of the Range of each Module's Degree, $[d^-, d^+]$.} Figure~\ref{fig:degree} illustrates the experimental results on different ranges, $[d^-, d^+]$, of each module's degree, from [1,9] to [1,5]. In Figure~\ref{subfig:degree_score}, the diversities of the new ring signatures that generated by our four approaches increase, when the average value of modules' degrees decreases. The reason is that, when the average value of modules' degree is higher, it is more difficult to satisfy the $\epsilon$-CIK constraint. In other words, the new ring signature has fewer valid modules. In Figure~\ref{subfig:degree_time}, the running time of our four approaches increases when the average value of modules' degree decreases. The reason is that, when the average value of modules' degree is higher, the new ring signature has fewer valid modules, which thus leads to lower complexity of the CIA-MS-DS problem and the decrease of the running time.

\noindent\textbf{Effect of the Range of each Module's $Pr_{max}$, $[PM^-, PM^+]$.} Figure~\ref{fig:pr} illustrates the experimental results on different ranges, $[PM^-, PM^+]$, of each module's $pr_{max}$, from [0.1,0.2] to [0.1,0.8]. In Figure~\ref{subfig:pr_score}, the diversities of the new ring signatures that generated by our four approaches decreases, when the range is wider. The reason is that, when the range of modules' size is wider, it is more difficult to satisfy the $\epsilon$-CIK constraint. In other words, the new ring signature has fewer valid modules. In Figure~\ref{subfig:pr_time}, the running time of our four approaches decreases when the difference of modules' $Pr_{max}$ increases. The reason is that, when the range of modules' $Pr_{max}$ is wider, the new ring signature has less valid modules, which thus leads to the lower complexity of the CIA-MS-DS problem and the decrease of the running time. Specifically, compared with the Game Algorithm, the Progressive Algorithm is more sensitive about the change of the range. The reason is that, when the range is wider, the cardinality of $M_{i,j}$ is smaller and the running time of the $\delta$-KP Algorithm is smaller. 

\begin{figure}[t!]\centering
	\subfigure{
		\scalebox{0.3}[0.3]{\includegraphics{figures/bar.png}}}\hfill
	\addtocounter{subfigure}{-1}
	\subfigure[][{\scriptsize Diversity}]{
		\scalebox{0.3}[0.3]{\includegraphics{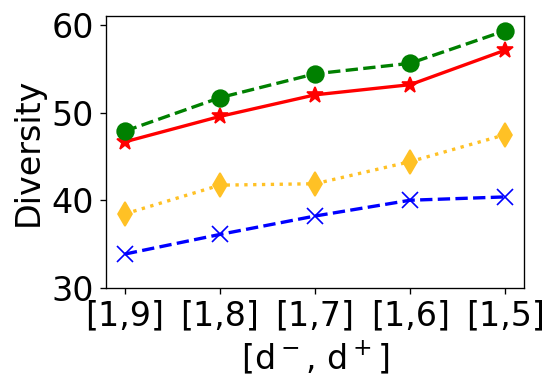}}
		\label{subfig:degree_score}}
	\subfigure[][{\scriptsize Running Time}]{
		\scalebox{0.3}[0.3]{\includegraphics{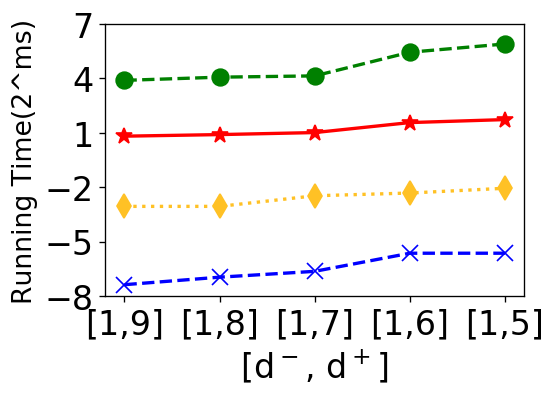}}
		\label{subfig:degree_time}}\figureCaptionMargin
	\caption{\small Effect of the Range of each Module's Degree (Synthetic)}\figureBelowMargin
	\label{fig:degree}
\end{figure}

\begin{figure}[t!]\centering
	\subfigure{
		\scalebox{0.3}[0.3]{\includegraphics{figures/bar.png}}}\hfill
	\addtocounter{subfigure}{-1}
	\subfigure[][{\scriptsize Diversity}]{
		\scalebox{0.33}[0.33]{\includegraphics{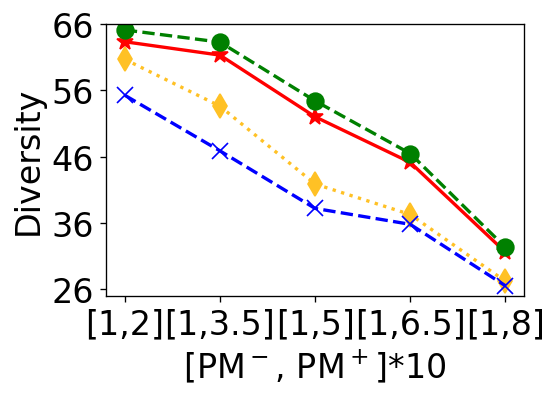}}
		\label{subfig:pr_score}}
	\subfigure[][{\scriptsize Running Time}]{
		\scalebox{0.33}[0.33]{\includegraphics{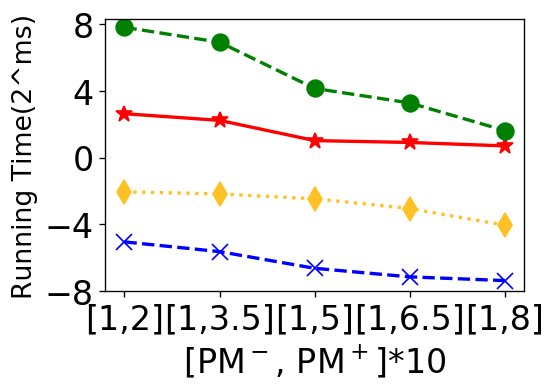}}
		\label{subfig:pr_time}}\figureCaptionMargin
	\caption{\small Effect of the Range of each Module's $Pr_{max}$ (Synthetic)}\figureBelowMargin
	\label{fig:pr}
\end{figure}
	
\end{document}